\documentclass[]{article}

\usepackage{amsmath}
\usepackage{amssymb}
\usepackage{amsthm}
\usepackage{tikz-cd}

\DeclareMathOperator*{\argmin}{arg\,min}

\usepackage{setspace}
\usepackage[margin=1in]{geometry}

\usepackage{graphicx}
\usepackage{xcolor}

\newtheorem{definition}{Definition}
\newtheorem{proposition}{Proposition}
\newtheorem{theorem}{Theorem}
\newtheorem{lemma}{Lemma}
\newtheorem{corollary}{Corollary}

\newtheorem{remark}{Remark}
\newtheorem{example}{Example}

\newcommand{\real}{\ensuremath{\mathbb{R}}}
\newcommand{\s}{\ensuremath{\mathbb{S}}}
\newcommand{\ltwo}{\ensuremath{\mathbb{L}^2}}
\newcommand{\cC}{\ensuremath{\mathbb{C}}}
\newcommand{\R}{\mathbb{R}}
\newcommand{\cE}{\mathcal{E}}
\newcommand{\cM}{\mathcal{M}}

\newcommand{\gauss}{\mathsf{G}}
\newcommand{\GM}{\mathsf{GM}}

\newcommand{\inner}[2]{\langle #1,#2 \rangle}

\usepackage{subfiles}

\title{A Wasserstein-type Distance for Gaussian Mixtures on Vector Bundles with Applications to Shape Analysis}
\author{Michael Wilson, Tom Needham, Chiwoo Park, Suprateek Kundu, and Anuj Srivastava}
\date{}

\begin{document}

\maketitle
	
	\begin{abstract}
	   This paper uses sample data to study the problem of comparing populations on finite-dimensional parallelizable Riemannian manifolds and more general trivial vector bundles. Utilizing triviality, our framework represents populations as mixtures of Gaussians on vector bundles and estimates the population parameters using a mode-based clustering algorithm. We derive a Wasserstein-type metric between Gaussian mixtures, adapted to the manifold geometry, in order to compare estimated distributions. Our contributions include an identifiability result for Gaussian mixtures on manifold domains and a convenient characterization of optimal couplings of Gaussian mixtures under the derived metric. We demonstrate these tools on some example domains, including the pre-shape space of planar closed curves, with applications to the shape space of triangles and populations of nanoparticles. In the nanoparticle application, we consider a sequence of populations of particle shapes arising from a manufacturing process, and utilize the Wasserstein-type distance to perform change-point detection.  
	\end{abstract}

	\section{Introduction} 
        Modern statistical analysis increasingly involves data objects that are nonlinear and non-Euclidean. A prominent example is {\it directional data}~\cite{mardia2000directional} where data naturally lies on a unit sphere. Another example is shape analysis, where one is interested in analyzing {\it shapes} of imaged objects. Although several approaches have been developed for shape analysis (see \cite{mardia-dryden-book,kendall-barden-carne,small-shapes,srivastava2016functional,younes-diffeo,bauer2022elastic} and others), they all agree in that the representation spaces of shapes are nonlinear. Examples of nonlinear data domains are also present in covariance analysis~\cite{dryden-covariance:2009,covariance-zhengwu:2018}, functional data analysis~\cite{srivastava2016functional}, and graphical data~\cite{jain2016geometry,chowdhury2020gromov,guo-JMIV:2022}. Analysis of non-Euclidean data requires statistical tools adapted to the differential geometries of the underlying representation spaces. These tools include statistical modeling, parameter estimation, and inferences. 
        Our paper is focused on a specific subproblem in this broad field, {\it i.e.},  comparing probability distributions on certain nonlinear domains. Precisely, we will model the probability distributions as {\it mixtures of Gaussians}, adapted to nonlinear domains of interest, and compare them using a novel variant of the Wasserstein distance. These choices -- mixtures of Gaussian models and Wasserstein metric -- are driven by convenience and applicability. Gaussian mixtures~\cite{delon-agnes:2019,takatsu:2008} provide a general yet parametric option for capturing population variability, and Wasserstein metrics have become a canonical choice for comparing distributions in a variety of contexts~\cite{villani2009,COTFNT}.

        We develop a general framework for domains that are trivial vector bundles. A \emph{vector bundle} is a set of isomorphic vector spaces indexed by points on a smooth manifold~\cite{milnor1974characteristic}; a vector bundle is called \emph{trivial} if this structure can be realized as a product of a smooth manifold with a vector space. The tangent bundle of a parallelizable manifold is an example of a trivial vector bundle. Examples of parallelizable manifolds include punctured spheres, spaces of symmetric positive definite matrices, Lie groups, and several shape spaces.
        
        Distributions on vector bundles provide a natural setting for combining results from the optimal transport~\cite{delon-agnes:2019,WGOT} and shape analysis~\cite{mardia-dryden-book,srivastava2016functional} literature. The two main goals of this paper are to develop a general theory for comparing certain distributions on trivial vector bundles (Section \ref{section:gaussians_on_vector_bundles}) and to apply this theory to study shape populations arising from images captured in a nano-manufacturing process (Section \ref{section: applications}).  The main challenges in deriving a framework for comparing probability distributions on vector bundles include (1)  the nonlinearity of underlying domains and the specification of convenient probability distributions for such domains, (2) efficient estimation of these distributions from data, and (3) comparisons of estimated distributions using proper metrics between distributions. We outline these choices next: 
         \begin{enumerate}
         
         \item  {\bf Forms of Probability Distributions}: Our first task is to define a probability distribution on a trivial bundle. While nonparametric approaches, often based on kernel methods, have gained prominence due to their generality and broad applicability, they require large sample sizes to capture the population variability effectively. In contrast, parametric families such as mixtures of Gaussians are robust under small sample sizes and have been covered extensively in the past literature. As mentioned earlier, the problem is complicated due to the nonlinearity of targeted domains. While some parametric families have been adapted from Euclidean to nonlinear domains, the choice is relatively limited. Some adaptations of Gaussians to nonlinear and compact domains include truncated Gaussians, von Mises, and wrapped Gaussian distributions~\cite{mardia2000directional}. This paper represents the underlying distribution as a mixture of Gaussians defined appropriately for vector bundles. (Note that modeling a population with a Gaussian mixture on vector bundles is akin to modeling that population as a wrapped Gaussian mixture on the base manifold.)  The primary motivation for choosing Gaussian mixtures is their generality, simplicity, and interpretability of the resulting Wasserstein distance.

        \item {\bf Estimating Probability Distributions}: The next issue is efficiently estimating Gaussian mixtures from given data. Depending on the chosen space and the Riemannian metric, several papers have studied the estimation of basic summary statistics from the data, such as means and covariances. However, the literature on estimating parameters of mixtures of Gaussians on nonlinear domains is relatively limited. The main issue is computational. The EM algorithm is an established approach for estimating Gaussian mixtures on vector spaces, but is expensive and prone to local solutions. Furthermore, its adaptations to nonlinear domains are costlier due to iterative computations of sample means~\cite{srivastava-joshi-etal:05,CHEN2021303}. Some papers have adapted EM algorithms for truncated Gaussians and von Mises distributions to nonlinear spaces (see \cite{Hauberg2018DirectionalSW}). We will modify and apply a recent method that performs clustering on manifold data by finding {\it modes} of the underlying distribution~\cite{deng-ISBI:2022}. We shall treat these modes as estimates of Gaussian means and further estimate covariances within individual clusters. The procedure for estimating cluster memberships, means, and covariances for general metric spaces is provided in~\cite{deng-ISBI:2022,deng-ICPR:2022}.

         \item  {\bf Metrics Between Probability Distributions}: The final issue is defining a metric for comparing and quantifying differences between chosen distributions in vector bundles. While there are several choices for this metric, the Wasserstein metric has become popular for several reasons. When available, it provides an interpretable solution for comparing probability distributions. It is also robust to misspecifications in distributions due to estimation errors. Finally, it leads to a closed form expression for comparing certain parametric families, in particular Gaussians. In this paper, we will utilize a Wasserstein-type metric previously developed for Euclidean domains~\cite{delon-agnes:2019}. (These are called Wasserstein-type because the couplings -- joint distributions for minimizing cost function -- are restricted to be mixtures of Gaussians, rather than all distributions.) Specifically, we will derive this metric for comparing mixtures of Gaussians on trivial vector bundles.

        A key idea that helps us define distances between Gaussians on trivial vector bundles is that any  trivialization leads to a consistent choice of basis for each of the tangent spaces of the manifold -- {\it i.e.}, a \emph{global moving frame}. Fixing consistent coordinates allows us to apply,  in a coherent manner, closed form expressions for distances between Gaussians on different tangent spaces. Later on we give examples of simple families of trivializations which are natural from an object-oriented data analysis persepective.
        \end{enumerate}

        We will demonstrate these ideas using both simulated and real-world datasets. In the simulated study, we will generate samples from mixtures of Gaussians on the (punctured) unit sphere $\s^2$ and demonstrate procedures for parameter estimation and population comparisons. We will also consider an example where the shape space of planar triangles~\cite{Kendall_triangle,cantarella2019random} is identified with $\s^2$, so that one can compare shape populations of triangles. In the real-data study, we investigate transmission-electron-microscopy videos of particles in nano-manufacturing processes, where each image frame contains hundreds of particles that differ in shape, size, and placement. We focus on the shapes of their contours and treat shapes in a frame as random samples from underlying shape populations. We model these shape populations as mixtures of Gaussians on the shape space of planar, closed contours. As mentioned above, we use a mode estimation procedure to infer the parameters of mixtures from the observed shape data for each frame separately. The goal is to track and compare temporal evolutions of these shape distributions and to quantify their changes over time. For instance, we use this quantification to detect change points in the manufacturing process.\\
    
    \noindent The salient contributions of this paper are as follows. 
        \begin{enumerate}
        \item It extends the notion of Gaussians and mixtures of Gaussians to trivial vector bundles, and uses them for statistical modeling and analysis. Examples of such domains include punctured spheres, tori, matrix Lie groups, spaces of symmetric positive-definite matrices, and other domains useful in statistical analysis. 
        
        \item It derives a convenient expression for comparing mixtures of Gaussians using a \\
        Wasserstein-type metric. This development provides useful insights into choices made for problem domain, probability model, and metric for comparing populations. 
       
         \item It applies these tools to comparisons of populations of planar contour shapes and for finding change points in temporal evolution of shape populations. 
        \end{enumerate}

        The paper proceeds as follows; in Section \ref{section: background}, we cover the background information necessary to introduce the Wasserstein-type distance for Gaussian mixtures on $\mathbb{R}^n$. In Section \ref{section:gaussians_on_vector_bundles}, we present our proposed framework for extending this Wasserstein-type distance to mixtures of Gaussians on vector bundles. In Section \ref{section: applications}, we present our experimental results involving real and simulated data. Section \ref{sec:conclusion} concludes the paper with some observations.

    \section{Background on Wasserstein Distances} \label{section: background}

        In this section, we introduce some background material and existing results to lay groundwork for our approach. Specifically, we focus on the mixtures of Gaussians on Euclidean spaces and the expressions for Wasserstein distances between such mixtures. 
        
        \subsection{Classical Wasserstein Distance}

        We begin with necessary background material on classical distances between probability distributions, called {\it Wasserstein distances}, on a general metric space.

        \paragraph{Wasserstein Distances for Metric Spaces}
        Let $(\mathcal{X},d)$ be a metric space.  
        
        \begin{definition}
        For $p \geq 1$, the \emph{Wasserstein space} $\mathcal{P}_p(\mathcal{X})$ is the set of probability measures on $\mathcal{X}$ with finite p-th moment, i.e., for every $x_0 \in \mathcal{X}$, the integral $\int_\mathcal{X} d(x_0,x)^p d\mu(x)$ is finite. The \emph{p-Wasserstein distance} $W^\mathcal{X}_p$ between probability measures $\mu_0,\mu_1 \in \mathcal{P}_p(\mathcal{X})$ is given by
        \begin{equation}\label{eq: w_p}
            W^\mathcal{X}_p(\mu_0, \mu_1) := \bigg(\inf_{\gamma \in \Pi(\mu_0, \mu_1)} \int_{\mathcal{X}\times\mathcal{Y}} d(x,y)^p d\gamma(x,y) \bigg)^\frac{1}{p}\ ,
        \end{equation}
        where $\Pi(\mu_0,\mu_1)$ is the set of \emph{couplings} of $\mu_0$ and $\mu_1$; that is, the set of joint probability measures $\gamma$ on $\mathcal{X} \times \mathcal{X}$ that have marginal distributions $\mu_0$ and $\mu_1$. A joint measure $\gamma$ that achieves the infimum of Equation~\eqref{eq: w_p} is called an \emph{optimal coupling}. 
        \end{definition} 
        
        The field of \emph{optimal transport (OT)} studies properties of the Wasserstein distance and related constructions; see, for example, \cite{COTFNT, Ambrosio2013, villani2009} for overviews of the well-developed theory of OT. In particular, the Wasserstein distance is a metric on $\mathcal{P}_p(\mathcal{X})$, under mild assumptions on $\mathcal{X}$ (e.g., $\mathcal{X}$ is a Polish space). 

        \paragraph{Finitely-Supported Measures}
        From an applications-oriented perspective, it is most common to consider the Wasserstein distance between finitely-supported distributions. In this setting, calculating the Wasserstein distance comes down to solving a constrained linear program.
        Indeed, for $i=0,1$, let
        \begin{equation*}
            \mu_i = \Sigma_{k=1}^{K_i} \alpha^k_{i} \delta_{x^k_{i}}, \ \Sigma_{k=1}^{K_i} \alpha^k_{i} = 1, \mbox{ where} \ \alpha^k_{i}>0 \ \forall \; i,k
        \end{equation*}
        be probability measures supported on points $x^k_{i} \in \mathcal{X}$. By an abuse of notation, we consider $\mu_i$ as a (column) vector $\mu_i = [\alpha^1_{i},\ldots,\alpha^{K_i}_k]^T \in \R^{K_i}$. Then the space of couplings can be identified with a set of matrices,
        \begin{equation*}
            \Pi(\mu_0, \mu_1) = \{ \pi \in \mathbb{R}^{K_0 \times K_1}: \pi^T\textbf{1} = \mu_0, \textbf{1}^T\pi = \mu_1^T \},
        \end{equation*}
        where $\textbf{1}$ always represents the column vector of all ones, whose size is inferred by context. When considering a coupling $\pi$ as a matrix, we write its $(i,j)$-entry as $\pi_{ij}$. Then the Wasserstein $p$-distance is given by
        \begin{equation*}
            W_p^\mathcal{X}(\mu_0, \mu_1)^p = \min_{\pi \in \Pi(\mu_0, \mu_1)} \Sigma_{i,j=1}^{K_0, K_1} \pi_{ij} d(x_0^{i}, x_1^{j})^p = \min_{\pi \in \Pi(\mu_0, \mu_1)} \langle D, \pi \rangle_F,
        \end{equation*}
        where $\langle \cdot, \cdot \rangle_F$ is the Frobenius inner product on $\R^{K_0 \times K_1}$ and $D \in \R^{K_0\times K_1}$ is the matrix with $(i,j)$-entry given by $d(x_0^{i}, x_1^{j})^p$. This shows that the objective of the Wasserstein distance computation is a linear function, and it is not hard to see that the constraint set $\Pi(\mu_0, \mu_1)$ is a convex polytope in $\R^{K_0 \times K_1}$. 

        \paragraph{Gaussian Distributions on $\real^d$}
        In general, calculating Wasserstein distances between continuous distributions is impossible due to the infinite-dimensional nature of the associated optimization problem. However, in the case of Gaussian distributions, there is a simple closed form equation for the Wasserstein distance in terms of the parameters of the distributions. We use $N_d(m, \Sigma)$ to denote the Gaussian distribution on $\R^d$ with mean $m \in \R^d$ and covariance $\Sigma \in \mathsf{Sym}_d^+$, where $\mathsf{Sym}_d^+ \subset \R^{d \times d}$ denotes the set of symmetric positive-definite matrices. When considering $\R^d$ as a metric space, we always use the standard Euclidean metric. The following result is classical.

        \begin{proposition}[See \cite{DOWSON, 10.1307/mmj/1029003026, MCCANN}]\label{prop:distance_between_gaussians}
            
        Given two Gaussian distributions on $\R^d$, $\eta_i = N_d(m_i, \Sigma_i), i \in \{0,1\}$, the squared 2-Wasserstein distance between $\mu_0$ and $\mu_1$ is given by
        
        \begin{equation}\label{eq: gaussian wasserstein}
        	W_2^{\R^d}({\eta_0},{\eta_1})^2 = \|m_0 - m_1\|^2 + \mathrm{tr}\left(\Sigma_0 + \Sigma_1 - 2(\Sigma_0^{\frac{1}{2}}\Sigma_1\Sigma_0^{\frac{1}{2}})^{\frac{1}{2}}\right).
        \end{equation}
        Moreover, an optimal coupling is given by a Gaussian measure on $\mathbb{R}^d \times \mathbb{R}^d$. If $m_0 = m_1 = 0 \in \R^d$, then there is an optimal Gaussian coupling with mean zero.
        \end{proposition}

        Let $\gauss_d := \{N_d(m,\Sigma): m \in \mathbb{R}^d, \Sigma \in \mathsf{Sym}_d^+ \}$ denote the set of Gaussian measures on $\R^d$. The above implies that $(\gauss_d,W_2)$ is a metric space whose metric is explicitly computable (here, we use $W_2 = W_2|_{\gauss_d \times \gauss_d}$, by abuse of notation). Due to this computational convenience, we focus on the $p=2$ version of Wasserstein distance for the rest of the paper.
        
        \subsection{Gaussian Mixture Measures on $\real^d$} 
        
        The closed formula \eqref{eq: gaussian wasserstein} for Wasserstein distance between Gaussians suggests that we consider a richer set of measures consisting of collections of Gaussians. More precisely:
        
        \begin{definition}\label{def:gaussian_mixture}
        	A measure $\mu$ on $\R^d$ is a \emph{Gaussian mixture measure} (or just \emph{Gaussian mixture}) if it can be written as
        	\begin{equation}\label{eqn:gaussian_mixture_defn}
        		\mu = \Sigma_{k=1}^K w_k\eta_k, \text{ where } \eta_k = N_d(m_k,\Sigma_k) \text{ and } \ \Sigma_{k=1}^K w_{k} = 1, \ w_k\geq0, \ \forall \; k.
        	\end{equation}

            A Gaussian mixture $\mu$ can also be considered as a discrete probability measure on $\gauss_d$. We use $\mu^\ast$ to distinguish this representation and write
            
            \begin{equation*}
        		\mu^* = \Sigma_{k=1}^K w_k\delta_{\eta_k}, \text{ where } \eta_k \in \gauss_d \text{ and } \ \Sigma_{k=1}^K w_{k} = 1, \ w_k\geq0 \ \forall \; k.
        	\end{equation*}

            We use $\GM_d$ to denote the collection of all Gaussian mixtures on $\R^d$. 
        \end{definition}

        An important property of Gaussian mixtures is that they are \emph{identifiable} in a certain precise sense, meaning that the representation given in \eqref{eqn:gaussian_mixture_defn} is essentially unique. Let us now explain this more precisely. The representation \eqref{eqn:gaussian_mixture_defn} is not strictly unique, as one could, for example, rearrange the terms or replace a term $w_k \eta_k$ by $(w_k/2)\eta_k + (w_k/2)\eta_k$ without changing the resulting measure. If a Gaussian mixture $\mu$ is written in the form \eqref{eqn:gaussian_mixture_defn} such that all Gaussians $\eta_k$ are pairwise distinct, we say that the representation is in \emph{minimal form}. We have the following classical result from \cite{yakowitz1968identifiability} (see also \cite[Proposition 2]{delon-agnes:2019}), which we record here for later use.

        \begin{proposition}[\cite{yakowitz1968identifiability}]\label{prop:identifiability}
            Let $\mu$ be a Gaussian mixture, with minimal form representations
            \[
            \sum_{k=1}^K w_k\eta_k \quad \mbox{ and } \quad \sum_{k=1}^{K'} w_k'\eta_k'.
            \]
            Then $K = K'$ and there exists a permutation $\sigma$ of $\{1,\ldots,K\}$ such that $w_k = w_{\sigma(k)}'$ and $\eta_k = \eta_{\sigma(k)}'$ for all $k$. 
        \end{proposition}

        The two perspectives on Gaussian mixture measures described in Definition \ref{def:gaussian_mixture} lead to two candidate metrics on $\GM_d$. On one hand, one could compute the Wasserstein distance $W_2^{\R^d}(\mu_0,\mu_1)$ between $\mu_0,\mu_1 \in \GM_d$. On the other hand, one could compute the Wasserstein distance in the metric space of discrete measures on $\gauss_d$, which reads as 
        \begin{equation}\label{eq: distance1}
            W_2^{\gauss_d}(\mu_0^*, \mu_1^*)^2 = \min_{\pi \in \Pi(w_0, w_1)} \Sigma_{k,l} \pi_{kl} W_2^{\R^d}(\eta_{0}^k, \eta_{1}^\ell)^2,
        \end{equation}
        where $\mu_i = \sum_k w_{i}^k \eta_{i}^k$ and $w_i = (w_{i}^k)_k$ for $i \in \{0,1\}$. This latter notion of distance between Gaussian mixtures was first studied in \cite{chen2018optimal}.

        In general, $W_2^{\R^d}(\mu_0,\mu_1)$ and $W_2^{\gauss_d}(\mu_0^\ast,\mu_1^\ast)$ are not equal. It turns out that this discrepancy can be reconciled by adding an extra constraint to the feasible set in the Wasserstein distance optimization problem. The following was first introduced in \cite{delon-agnes:2019}.

        \begin{definition}[\cite{delon-agnes:2019}]\label{def:mixture_wasserstein distance}
        Given $\mu_i \in \GM_d$, $i \in \{0,1\}$, the \emph{mixture Wasserstein distance} is given by
        \begin{equation*}
        	MW_2^{\R^d}(\mu_0, \mu_1)^2 := \inf_{\pi \in \Pi(\mu_0,\mu_1)\cap \GM_{2d}} \int_{\mathbb{R}^d \times \mathbb{R}^d} \|x - y\|^2 \ d\pi(x,y).
        \end{equation*}
        \end{definition}

        It is shown in \cite[Proposition 4]{delon-agnes:2019} that the alteration of the Wasserstein distance given in Definition \ref{def:mixture_wasserstein distance} agrees with the distance  between Gaussian mixtures used in Equation \eqref{eq: distance1}. We record this result here:

        \begin{theorem}[\cite{delon-agnes:2019}]\label{thm:Gaussian_mixture_metrics}
        For $\mu_i \in \GM_d$, $i \in \{0,1\}$, we have
        $
            MW_2^{\R^d}(\mu_0, \mu_1) = W_2^{\gauss_d}(\mu_0^*, \mu_1^*).
        $
        \end{theorem}

    \section{Distances Between Gaussian Mixtures on Vector Bundles}\label{section:gaussians_on_vector_bundles}

    The main contribution of this paper is to generalize the Wasserstein-type distance $MW_2^{\R^d}$ to Gaussian Mixtures defined on trivial vector bundles. We first introduce some preliminary ideas.

    \subsection{Preliminaries on Gaussian Mixtures on Vector Bundles}
    
    \paragraph{Gaussian Measures on Inner Product Spaces} Let $(V,\langle \cdot, \cdot \rangle_V)$ be a finite-dimensional inner product space. We wish to consider Gaussian measures on $V$. These could, of course, be defined by choosing an isometry to $\R^d$ and transferring the standard definition. It will be convenient for later computations to have a more coordinate-free description of Gaussian measures on $V$. We develop this point of view here.

    \begin{definition}\label{def:multivariate_gaussian}
        A Borel measure $\mu$ on $V$ is called \emph{Gaussian} if, for every linear functional $f:V \to \R$, the pushforward $f_\# \mu$ is a Gaussian measure (in the standard sense) on $\R$.
    \end{definition} 

    Definition \ref{def:multivariate_gaussian} is used in  \cite[Definition 1.2.1]{bogachev1998gaussian} to characterize Gaussian measures on Euclidean spaces. Indeed, it is straightforward to show that any Gaussian $N_d(m,\Sigma)$ on $\R^d$ (in the sense of the previous section) has this pushforward property. Using this definition for more general inner product spaces allows us to easily relate Gaussian measures to their Euclidean counterparts. 

    \begin{proposition}
        A Borel measure $\mu$ on $(V,\langle \cdot, \cdot \rangle_V)$, where $\mathrm{dim}(V) = d$, is Gaussian if and only if there is a Gaussian measure $\nu$ on $\R^d$ and a linear isometry $g:\R^d \to V$ such that $\mu = g_\# \nu$.
    \end{proposition}

    \begin{proof}
        First assume that $\mu$ is Gaussian, let $h:V \to \R^d$ be an arbitrary linear isometry and set $\nu = h_\# \mu$. Then for any linear functional $f:\R^d \to \R$, we have
        $f_\#\nu = f_\# (h_\# \mu) = (f \circ h)_\# \mu$ and $f \circ h: V \to \R$ is a linear functional. It follows that $f_\# \nu$ is Gaussian on $\R$, so that $\nu$ must be Gaussian on $\R^d$. Setting $g = h^{-1}$, we have that $\mu = g_\# \nu$.

        Conversely, suppose that $\mu = g_\# \nu$ for a Gaussian $\nu$ on $\R^d$. Then for any linear functional $f:V \to \R$, we have that $f_\# \mu = f_\# g_\# \nu = (f \circ g)_\# \nu$ is a Gaussian on $\R$. It follows that $\mu$ is Gaussian on $V$.
    \end{proof}

    Given a Gaussian $\nu = N_d(m,\Sigma)$ on $\R^d$ and a linear isometry $g:\R^d \to V$, let $\mu = g_\# \nu$. We would like to define the mean and covariance of $\mu$ to be $g(m)$ and $g \Sigma g^{-1}$, respectively. Here, we consider $\Sigma$ as an operator $\R^d \to \R^d$. We first show that these quantities do not depend on the choice of $\nu$ or $g$. 

    \begin{proposition}\label{prop:covariance_well_defined}
        Let $\nu = N_d(m,\Sigma)$ and $\nu' = N_d(m',\Sigma')$ be Gaussians on $\R^d$ and let $g$ and $g'$ be linear isometries from $\R^d$ to $V$ such that $g_\# \nu = g'_\# \nu'$. Then $g(m) = g'(m')$ and $g \Sigma g^{-1} = g' \Sigma' (g')^{-1}$.
    \end{proposition}

    \begin{proof}
        We will show that the quantities $g(m)$ and $\widetilde{\Sigma} := g^{-1} \Sigma g$ are intrinsic to the measure $\mu := g_\# \nu = g'_\# \nu'$, whence the claims will follow. We have
        \[
        \int_V v \;d\mu(v) = \int_V v \;d(g_\# \nu) (v) = \int_{\R^d} g(w) \;d\nu(w) = g\left(\int_{\R^d} w \;d\nu(w) \right) = g(m),
        \]
        where the second equality is the change of variables formula and third follows by the assumption that $g$ is an isometry. Next, we have 
        \begin{align}
            \int_V \langle u, x-\tilde{m} \rangle_V \langle v, x-\tilde{m} \rangle_V \; d\mu(x) &= \int_{\R^d} \langle g(s), g(y) - g(m) \rangle_V \langle g(t), g(y) - g(m) \rangle_V \; d\nu(y) \label{eqn:invariance_proof_1} \\
            &= \int_{\R^d} \langle s, y - m \rangle \langle t, y - m \rangle \; d\nu(y) \label{eqn:invariance_proof_2} \\
            &= \langle \Sigma s, t \rangle \label{eqn:invariance_proof_3} \\
            &= \langle g^{-1} \widetilde{\Sigma} g g^{-1} u, g^{-1} v \rangle = \langle \widetilde{\Sigma}u, v \rangle_V. \nonumber
        \end{align}
        Equation \eqref{eqn:invariance_proof_1} is the change of variables formula with $s:= g^{-1}(u)$ and $t := g^{-1}(v)$, \eqref{eqn:invariance_proof_2} uses the fact that $g$ is a linear isometry, \eqref{eqn:invariance_proof_3} is \cite[Corollary 1.2.3]{bogachev1998gaussian}, and the remaining equalities follow by definition and the fact that $g^{-1}$ is an isometry. These identities give the desired intrinsic characterizations.
    \end{proof}

    As a corollary, we get that the following is a valid definition.

    \begin{definition}\label{def:gaussian_inner_product}
        Let $(V,\langle \cdot, \cdot \rangle_V)$ be an inner product space, $g:\R^d \to V$ a linear isometry and $\nu = N_d(m, \Sigma) \in \gauss_d$. The \emph{mean} of the Gaussian measure $g_\# \nu$ on $V$ is $g(m)$ and the \emph{covariance operator} is $g \Sigma g^{-1}$. 
    \end{definition}

    \paragraph{Gaussian Mixtures on Vector Bundles} Let $\mathcal{M}$ be a Riemannian manifold. When dealing with data valued in such a manifold, it is common to \emph{linearize} the analysis by choosing a basepoint $m \in \mathcal{M}$ and pulling the data back to the tangent space $T_m \mathcal{M}$ via the Riemannian log map; for example, this is a standard technique in statistical shape analysis \cite{srivastava2016functional, bharath2018radiologic} and computational optimal transport \cite{wang2013linear,chowdhury2020gromov}. From a modeling perspective, it is often useful to fit a distribution to the linearized data, resulting in a probability distribution on $T_m \mathcal{M}$. One can consider the resulting distribution as a highly singular probability measure on the tangent bundle $T\mathcal{M}$, in the sense that it is only supported on the fiber $T_m \mathcal{M} \subset T\mathcal{M}$. This is the perspective taken in~\cite{srivastava2016functional,WGOT}, where the authors consider Gaussian distributions on tangent spaces as models for ``wrapped Gaussians" on the underlying manifold (this terminology originates in the directional statistics literature---see \cite{collett1981discriminating,mardia2000directional}). Observe that the well definededness of this framework depends on technicalities such as the domain of the log map---this hints at the utility of considering parallelizable manifolds, as we do in the sequel.

    In this paper, we propose a data model which linearizes subsets of the data at various strategically chosen basepoints, $m_1,\ldots,m_K \in \mathcal{M}$. Fitting (weighted) distributions on each $T_{m_k} \mathcal{M}$ leads to a more general singular measure on $T\mathcal{M}$ whose support is contained in $\cup_k T_{m_k} \mathcal{M} \subset T\mathcal{M}$. Arguably the simplest such model involves fitting mean zero Gaussian distributions in each tangent space. 
    
    We now formalize the concepts described above. It will be convenient to work, more generally, in the setting of vector bundles. Let $p:\mathcal{E} \to \mathcal{M}$ be a rank-$d$ vector bundle over a smooth manifold $\mathcal{M}$; in what follows, we typically denote the vector bundle as $\cE \to \cM$, with the understanding that there is an underlying projection map that has been supressed from the notation. We denote the fiber over $m \in \mathcal{M}$ as $\mathcal{E}_m \approx \R^d$. Let $\langle \cdot, \cdot \rangle = \{\langle \cdot, \cdot \rangle_m\}_{m \in M}$ be a smoothly-varying family of inner products on the fibers $\mathcal{E}_m$. 
    
    \begin{definition}\label{def:gaussian_mixture_vector_bundle}
        A Borel measure $\eta$ on the vector bundle $\mathcal{E}$ is called a \emph{Gaussian measure} if it is a mean-zero Gaussian measure on the inner product space $(\mathcal{E}_m,\langle \cdot, \cdot \rangle_m)$, for some $m \in \cM$ (see Definition \ref{def:gaussian_inner_product}). If the covariance operator of $\eta$ is $\Sigma$, we write $\eta = N_\mathcal{E}(m,\Sigma)$. The collection of Gaussian measures on $\mathcal{E}$ is denoted $\mathsf{G}(\mathcal{E})$. 

        A Borel measure $\mu$ on $\mathcal{E}$ is called a 
        \emph{Gaussian mixture measure} (or just \emph{Gaussian mixture}) if it can be written as $\mu = \Sigma_{k=1}^K w_k \eta_k$, where each $\eta_k = N_{\mathcal{E}}(m_k,\Sigma_k)$ for some $m_k \in \mathcal{M}$, and where $\Sigma_{k=1}^K w_{k} = 1$. We denote the collection of all Gaussian mixtures on $\mathcal{E}$ as $\GM(\mathcal{E})$. 
    \end{definition}

    As an important example, consider a product bundle $\mathcal{E} = \mathcal{M} \times \R^d$, where $\mathcal{E}_m = \{m\} \times \R^d \approx \R^d$ is endowed with the standard inner product. Then a Gaussian mixture on $\mathcal{E}$ is simply a collection of mean-zero Gaussians indexed by a finite collection of points $\{m_1,\ldots,m_K\}$ in $\mathcal{M}$. In particular, the following result is immediate.

    \begin{proposition}\label{prop:gaussian_mixtures_correspondence}
        We have $\gauss_d \approx \gauss(\R^d \times \R^d)$ and $\GM_d \approx \GM(\R^d \times \R^d)$, as sets. To make this more precise, let $\eta = N_d(m,\Sigma)$ be a Gaussian measure on $\R^d$ and let $\overline{\eta}$ denote the measure when considered as a Gaussian measure on the trivial bundle $\mathcal{E} = \R^d \times \R^d$; that is, $\overline{\eta} = N_\mathcal{E}(m,\Sigma)$. The map $\eta \mapsto \overline{\eta}$ induces a bijective correspondence between Gaussian mixture measures $\GM_d$ on $\R^d$ (in the sense of Definition \ref{def:gaussian_mixture}) and Gaussian mixture measures $\GM(\mathcal{E})$ on $\mathcal{E}$ (in the sense of Definition \ref{def:gaussian_mixture_vector_bundle}). Explicitly, the bijection maps $\mu \in \GM_d$, written in minimal form as $\sum_k w_k \eta_k$, to $\sum_k w_k \overline{\eta}_k \in \GM(\mathcal{E})$. 
    \end{proposition}

    We now extend the discussion of identifiability of Gaussian mixture measures (see Proposition \ref{prop:identifiability}) to the setting of vector bundles. In analogy with the Euclidean setting, if a Gaussian mixture measure $\mu$ on $\mathcal{E}$ is written as $\mu = \sum_k w_k \eta_k$, we say that the representation is in \emph{minimal form} if the measures $\eta_k$ are pairwise distinct. 

    \begin{proposition}\label{prop:identifiability_vector_bundles}
        Let $\mu \in \GM(\mathcal{E})$ be a Gaussian mixture on a vector bundle $\mathcal{E}$ with minimal form representations
        \[
        \sum_{k=1}^K w_k \eta_k \quad \mbox{ and } \quad \sum_{k=1}^{K'} w_k' \eta_k'.
        \]
        Then $K = K'$ and there exists a permutation $\sigma$ of $\{1,\ldots,K\}$ such that $w_k = w_{\sigma(k)}'$ and $\eta_k = \eta_{\sigma(k)}'$ for all $k$. 
    \end{proposition}

    \begin{proof}
        Because $\mu$ is a Gaussian mixture measure, its support is of the form $\cup_j^J \mathcal{E}_{m_j}$ for some pairwise distinct points $m_j \in \mathcal{M}$. Fix $m = m_j$ and consider the restriction of $\mu|_{\mathcal{E}_m}$ of $\mu$ to $\mathcal{E}_m$. There are some subcollections $\{\eta_{k_i}\}_{i=1}^{I}$ and $\{\eta_{k_i}'\}_{i=1}^{I'}$ of measures which are Gaussians on $\mathcal{E}_{m}$. Let us assume without loss of generality that $(\mathcal{E}_m, \langle \cdot, \cdot \rangle_m) = (\R^d,\langle \cdot, \cdot \rangle)$ (the latter endowed with the standard inner product)---the two inner product spaces are isometric, so this assumption can be made without loss of generality, allowing us to suppress the isometry from the notation. Then the measures
        \[
        \sum_{i=1}^I \left(\frac{w_{k_i}}{\mu(\mathcal{E}_m)}\right) \eta_{k_i} \quad \mbox{ and } \quad \sum_{i=1}^{I'} \left(\frac{w_{k_i}'}{\mu(\mathcal{E}_m)}\right) \eta_{k_i}'
        \]
        are representations of the Gaussian mixture measure $\frac{1}{\mu(\mathcal{E}_m)}\mu|_{\mathcal{E}_m}$ on $\R^d$ which are in minimal form. It follows from Proposition \ref{prop:identifiability} that $I = I'$ and that the $w_{k_i}$ and $\eta_{k_i}$ agree with the $w_{k_i}'$ and $\eta_{k_i}'$ up to a permutation of $\{1,\ldots,I\}$. Running the same argument on each $m_j$ completes the proof.
    \end{proof}

    We also have the following immediate corollary, which will be useful later on.
    
    \begin{corollary}\label{cor:identifiability_vector_bundles}
        Let $\mu \in \GM(\cE)$ with not necessarily minimal form representations $\sum_{j=1}^K w_j \eta_j$ and $\sum_{k=1}^{K'} w_k' \eta_k'$. Then for every $j \in \{1,\ldots,K\}$, there exists $k \in \{1,\ldots,K'\}$ such that $\eta_k' = \eta_j$. 
    \end{corollary}

    \paragraph{Gaussian Mixtures on Trivial Bundles}\label{sub:trivial_bundles} From now on, we restrict our attention to the especially simple case of \emph{trivial} vector bundles; we recall the definition here. Vector bundles $\mathcal{E} \to \cM$ and $\cE' \to \cM$ over the same base space are called \emph{isomorphic} if there exists a diffeomorphism $\varphi:\mathcal{E} \to \cE'$ such that the diagram
    \[\begin{tikzcd}
    	\mathcal{E} && {\cE'} \\
    	& \cM
    	\arrow[from=1-1, to=2-2]
    	\arrow[from=1-3, to=2-2]
    	\arrow["\varphi", from=1-1, to=1-3]
    \end{tikzcd}\]
    commutes, where the diagonal arrows are vector bundle projections, and such that the induced maps on fibers $\varphi_m:=\varphi|_{\mathcal{E}_m}$ are linear isomorphisms $\mathcal{E}_m \to \cE'_m$ for each $m \in \mathcal{M}$. The map $\varphi$ is called a \emph{bundle isomorphism}. If $\cE$ and $\cE'$ are both endowed with smoothly-varying inner products and each $\varphi_m$ is an isometry of inner product spaces, then we say that $\varphi$ is a \emph{bundle isometry}. We now consider rank-$d$ vector bundles $\cE \to \cM$ which are isomorphic to the product bundle $\cM \times \R^d$; such bundles are called \emph{trivial}.  In this case, an isomorphism $\varphi:\cE \to \cM \times \R^d$ is called a \emph{trivialization} of $\mathcal{E}$. We will use the following basic result, which says that we can assume without loss of generality that trivializations are bundle isometries with respect to the standard structure on $\cM \times \R^d$.
    
    \begin{proposition}
        If $\mathcal{E} \to \cM$ is a rank-$d$ trivial bundle endowed with a smoothly-varying inner product $\{\langle \cdot, \cdot\rangle_m\}_{m \in \cM}$, we can choose a trivialization  which is a bundle isometry with respect to the standard inner product on $\R^d \approx \{m\} \times \R^d$.
    \end{proposition} 

    \begin{proof}
        First observe that any smooth inner product on $\cM \times \R^d$ is bundle isometric to the standard one. Indeed, this is achieved by choosing a smoothly-varying orthonormal basis (with respect to the arbitrary inner product) for each fiber $\{m\} \times \R^d$, then defining the bundle isomorphism by sending this to the standard orthonormal basis. Now, for an arbitrary trivialization $\varphi:\cE \to \cM \times \R^d$, define an inner product on $\cM \times \R^d$ by pulling back each $\langle \cdot, \cdot\rangle_m$ via $\varphi^{-1}$. Choose a bundle isometry $\psi:\cM \times \R^d \to \cM \times \R^d$ sending this pullback family of inner products to the standard one. Then $\psi \circ \varphi:\cE \to \cM \times \R^d$ is a bundle isometry.
    \end{proof}
    
    We justify our interest in the class of trivial vector bundles by the following remarks, elaborating on the discussion in the introduction.

    \begin{remark}
        \begin{enumerate}
            \item  From an applications-oriented perspective, we are especially interested in \emph{parallelizable Riemannian manifolds}. That is, the case where $\mathcal{M}$ is a Riemannian manifold and the trivial vector bundle $\mathcal{E}$ is the tangent bundle $T\mathcal{M}$. For example, any 3-dimensionsonal manifold is parallelizable \cite{benedetti2019framing}, as is any Lie group \cite{lee2012smooth}.

            Although we frequently work with manifolds $\mathcal{M}$ which are not parallelizable (e.g., spheres $\s^n$ with $n \not \in \{0,1,3,7\}$---see \cite{bott1958parallelizability}), it is typically the case in realistic data analysis applications that manifold-valued data lies in a subset $\widetilde{\mathcal{M}} \subset \mathcal{M}$ which is parallelizable, so that we may assume parallelizability without loss of generality. In the setting of the sphere $M = \s^n$, any proper open subset $\widetilde{\mathcal{M}} \subset M$ is parallelizable---indeed, a proper open subset must miss a point $p \in \mathcal{M}$ and $\widetilde{\mathcal{M}} = \mathcal{M} \setminus \{p\}$ is parallelizable, and any open submanifold of a parallelizable manifold is also parallelizable. 

            Thus, it frequently suffices to consider trivial vector bundles when focused on applications.

            \item In \cite{WGOT}, a Wasserstein distance is defined between Gaussian measures on the tangent bundle $T\cM \to \cM$ of a Riemannian manifold $\cM$. This is done with respect to an extended metric (i.e., a metric-like function which is allowed to take the value $\infty$) on $T\cM$, which is a true (finite) metric only in the case that $\cM$ has empty cut locus. This implies that $\cM$ is diffeomorphic to Euclidean space, and, in particular, that $T\cM$ is a trivial vector bundle. Our setting is therefore a generalization, in a formal sense, as it includes tangent bundles where the base manifold can have nontrivial topology (e.g., $\s^1$). In practice, the present setting and that of \cite{WGOT} are equivalent in the case of Gaussian measures on tangent bundles, but we further generalize to the novel case of Gaussian mixtures.
            
            \item As shown in Proposition \ref{prop:gaussian_mixtures_correspondence}, Gaussian mixtures on $\R^d$ (in the classical sense) correspond to Gaussian mixtures on the trivial vector bundle $\R^d \times \R^d$. The setting of trivial bundles is therefore a natural generalization of the classical setting.
        \end{enumerate}
    \end{remark}

    The next result shows that the property of a measure being a Gaussian mixture is well-defined on bundle isometry classes.
    
    \begin{proposition}\label{prop:gaussians_bundle_isomorphisms}
        Let $\cE,\cE' \to \cM$ be vector bundles endowed with inner products. If $\varphi:\cE \to \cE'$ is a bundle isometry then a measure $\mu$ on $\cE$ is a Gaussian mixture if and only if its pushforward $\varphi_\# \mu$ is a Gaussian mixture on $\cE'$. 

        Explicitly, if $\mu = \sum_{k=1}^K w_k \eta_k$ with $\eta_k = N_\cE(m_k,\Sigma_k)$, then  
        $
        \varphi_\# \mu = \sum_{k=1}^K w_k \varphi_\# \eta_k,
        $
        where $\varphi_\# \eta_k = N_{\cE'}(\varphi(m_k),\varphi_{m_k} \Sigma_k \varphi_{m_k}^{-1})$.
    \end{proposition}

    \begin{proof}
        We have
        \[
        \varphi_\# \mu = \sum_k w_k \varphi_\# \eta_k,
        \]
        by the linearity of pushforwards. Moreover, each $\varphi_\# \eta_k$ is a Gaussian on $\cE'$ with mean $\varphi(m_k)$, since for any linear functional $f:\cE'_{m_k} \to \R$, we have
        \[
        f_\# \big(\varphi_\# \eta_k\big) = \big(f \circ \varphi_{m_k}\big)_\# \eta_k,
        \]
        and $f \circ \varphi_{m_k}$ is a linear functional on $\mathcal{E}_{m_k}$, so the result must be a Gaussian on $\R$. It remains to derive the formula for the covariance operator of the pushforward. Choose a linear isometry $g:\R^d \to \cE_{m_k}$ and a covariance operator $\widetilde{\Sigma} \in \mathsf{Sym}_d^+$ such that $\Sigma= g \widetilde{\Sigma} g^{-1}$ (see Proposition \ref{prop:covariance_well_defined} and Definition \ref{def:gaussian_inner_product}). Then $\varphi_{m_k} \circ g :\R^d \to \cE'_{\varphi(m_k)}$ is a linear isometry taking the covariance operator $\widetilde{\Sigma}$ to 
        \[
        \big(\varphi_{m_k} \circ g\big) \widetilde{\Sigma} \big(\varphi_{m_k} \circ g\big)^{-1} = \varphi_{m_k} \Sigma \varphi_{m_k}^{-1}.
        \]

        The converse follows by the same argument, using $\varphi^{-1}$ in place of $\varphi$.
    \end{proof}

    \subsection{Distances Between Gaussian Mixtures on Trivial Bundles}\label{sec:distance_gaussian_mixtures_trivial}

    We now extend Theorem \ref{thm:Gaussian_mixture_metrics} to the setting of Gaussian mixtures on trivial vector bundles.
    
    \paragraph{Distance Between Gaussians on Trivial Bundles} Next, we characterize distances between Gaussian measures on a trivial bundle with respect to a family of metrics. In the following, suppose that we have some fixed metric $d_\cM$ on our base manifold $\cM$ (e.g., geodesic distance with respect to a Riemannian metric). 
    
    For a product bundle $\cM \times \R^d \to \cM$, let $d_{\cM \times \R^d}:(\cM \times \R^d) \times (\cM \times \R^d) \to \R$ be the $\ell_2$-product metric with respect to $d_\cM$ and Euclidean distance. That is,
    \[
    d_{\cM \times \R^d}((m_0,v_0),(m_1,v_1))^2 = d_\cM(m_0,m_1)^2 + \|v_0 - v_1\|^2.
    \]
    Let $W_2^{\cM \times \R^d}$ denote the associated 2-Wasserstein distance. We will use the following function, which is clearly a smooth bijection.
    
    \begin{definition}\label{def:coordinate_permutation_map}
        The \emph{coordinate permutation map} is
    \begin{align*}
        \sigma:(\cM \times \cM) \times (\R^d \times \R^d) &\to (\cM \times \R^d) \times (\cM \times \R^d) \\
        ((p_0,p_1),(v_0,v_1)) &\mapsto ((p_0,v_0),(p_1,v_1)).
    \end{align*}
    \end{definition}

    \begin{lemma}\label{lem:gaussian_distance_product}
        Let $\eta_i = N_{\cM \times \R^d}(m_i,\Sigma_i)$ be Gaussian measures on the product bundle $\cM \times \R^d$ for $i \in \{0,1\}$. Then
        \[
        W_2^{\cM \times \R^d}(\eta_0,\eta_1) = d_\cM(m_0,m_1)^2 + \mathrm{tr}\big(\Sigma_0 + \Sigma_1 - 2 \big(\Sigma_0^{\frac{1}{2}} \Sigma_1 \Sigma_0^{\frac{1}{2}}\big)^\frac{1}{2} \big).
        \]
        Moreover, there is an optimal coupling of the form $\sigma_\# \overline{\pi}$ for some $\overline{\pi} \in \gauss((\cM \times \cM) \times (\R^d \times \R^d))$, where $\sigma$ is the coordinate permuatation map.
    \end{lemma}

    \begin{proof}
        Any coupling $\pi \in \Pi(\eta_0,\eta_1)$ must be supported on the product of the supports of the $\eta_i$; that is, 
        \[
        \mathrm{supp}(\pi) \subset (\{m_0\} \times \R^d) \times (\{m_1\} \times \R^d) \approx \R^d \times \R^d.
        \]
        We then have
        \begin{align}
            W_2^{\cM \times \R^d}(\eta_0,\eta_1)^2 &= \inf_{\pi \in \Pi(\eta_0,\eta_1)} \int_{(\cM \times \R^d) \times (\cM \times \R^d)} d_{\cM \times \R^d}((p_0,v_0),(p_1,v_1))^2 d\pi((p_0,v_0),(p_1,v_1)) \nonumber \\
            &= \inf_{\pi \in \Pi(\eta_0,\eta_1)} \int_{(\{m_0\} \times \R^d) \times (\{m_1\} \times \R^d)} d_\cM(m_0,m_1)^2 \nonumber \\
            &\qquad \qquad  + \int_{(\{m_0\} \times \R^d) \times (\{m_1\} \times \R^d)} \|v_0 - v_1\|^2 d\pi((m_0,v_0),(m_1,v_1)) \nonumber \\ 
            &= d_\cM(m_0,m_1)^2 \label{eqn:gaussian_distance_1} \\
            &\qquad \qquad + \inf_{\pi \in \Pi(\eta_0,\eta_1)} \int_{(\{m_0\} \times \R^d) \times (\{m_1\} \times \R^d)} \|v_0 - v_1\|^2 d\pi((m_0,v_0),(m_1,v_1)). \nonumber
        \end{align}
        Now consider the bijection
        \begin{align*}
            \rho: (\{m_0\} \times \R^d) \times (\{m_1\} \times \R^d) &\to \R^d \times \R^d \\
            ((m_0,v_0),(m_1,v_1)) &\mapsto (v_0,v_1).
        \end{align*}
        We claim that this induces a correspondence
        \[
        \rho_\#:\Pi(\eta_0,\eta_1) \to \Pi(N_d(0,\Sigma_0),N_d(0,\Sigma_1)).
        \]
        Indeed, for $\pi \in \Pi(\eta_0,\eta_1)$ and any measurable subset $A \subset \R^d$, we have 
        \[
        \rho_\# \pi (A \times \R^d) = \pi(\rho^{-1}(A \times \R^d)) = \pi((\{m_0\} \times A) \times \{m_1\} \times \R^d) = \eta_0(\{m_0\} \times A) = N_d(0,\Sigma_0)(A),
        \]
        and the computation for the other marginal is similar. Combining this observation with the change-of-variables formula yields 
        \begin{align*}
        &\inf_{\pi \in \Pi(\eta_0,\eta_1)} \int_{(\{m_0\} \times \R^d) \times (\{m_1\} \times \R^d)} \|v_0 - v_1\|^2 d\pi((m_0,v_0),(m_1,v_1)) \\
        &\qquad \qquad \qquad = \inf_{\pi \in \Pi(\eta_0,\eta_1)} \int_{\R^d \times \R^d} \|v_0 - v_1\|^2 d(\rho_\# \pi)(v_0,v_1) \\
        &\qquad \qquad \qquad = \inf_{\xi \in \Pi(N_d(0,\Sigma_0),N_d(0,\Sigma_1))} \int_{\R^d \times \R^d} \|v_0 - v_1\|^2 d\xi(v_0,v_1).
        \end{align*}
        Thus the second term of \eqref{eqn:gaussian_distance_1} is equivalent to computing the Wasserstein distance between mean-zero Gaussians in $\R^d$, so Proposition \ref{prop:distance_between_gaussians} implies that the claimed formula for \newline $W_2^{\cM \times \R^d}(\eta_0,\eta_1)$ holds. Moreover, the proposition tells us that there is an optimal coupling $\xi$ which is a Gaussian on $\R^d \times \R^d$ with mean zero. The pushforward $(\rho^{-1})_\# \xi \in \Pi(\eta_0,\eta_1)$ is an optimal coupling for the original Wasserstein distance.
        To prove the last statement of the proposition, first consider the measurable map
        \begin{align*}
            \tau: \R^d \times \R^d &\to (\mathcal{M} \times \cM) \times (\R^d \times \R^d) \\
            (v_0,v_1) &\mapsto ((m_0,m_1),(v_0,v_1)).
        \end{align*}
        It is easy to see that $\tau_\#$ takes mean zero Gaussians on $\R^d \times \R^d$ to Gaussians on the product bundle $(\cM \times \cM) \times (\R^d \times \R^d)$; explicitly,
        \[
        \tau_\# N_{2d}(0,\Sigma) = N_{(\cM \times \cM) \times (\R^d \times \R^d)}((m_0,m_1),\Sigma).
        \]
        Next, observe that $\rho^{-1} = \sigma \circ \tau$, so it follows that the optimal coupling $(\rho^{-1})_\# \xi$ can be expressed as $\sigma_\# \overline{\pi}$, where $\overline{\pi} := \tau_\# \xi \in \gauss((\cM \times \cM) \times (\R^d \times \R^d))$. 
    \end{proof}

    This result extends to the setting of a general trivial bundle (not just a product bundle) $\cE \to \cM$. Let $\varphi:\cE \to \cM \times \R^d$ be a trivialization and define a metric $d_{\varphi}$ on $\cE$ as the pullback of the product metric; that is,
    \[
    d_{\varphi}(q_0,q_1) := \varphi^\ast d_{\cM \times \R^d}(q_0,q_1) = d_{\cM \times \R^d}(\varphi(q_0),\varphi(q_1)).
    \]
    Let $W_2^\varphi$ denote the Wasserstein distance associated to this metric. 

    \begin{proposition}\label{prop:trivial_bundle_gaussian_distance}
        Let $\cE \to \cM$ be a trivial bundle with $\varphi$ and $W_2^\varphi$ as above and let $\eta_i = N_\cE(m_i,\Sigma_i)$ for $i \in \{0,1\}$. Then 
        \[
        W_2^{\varphi}(\eta_0,\eta_1) = d_\cM(m_0,m_1)^2 + \mathrm{tr}\big(\Sigma_0 + \Sigma_1 - 2 \big(\Sigma_0^{\frac{1}{2}} \Phi_{m_0,m_1}^{-1} \Sigma_1 \Phi_{m_0,m_1} \Sigma_0^{\frac{1}{2}}\big)^\frac{1}{2} \big),
        \]
        where
        \[
        \Phi_{m_0,m_1} := \varphi_{m_1}^{-1} \circ \varphi_{m_0} : \cE_{m_0} \to \cE_{m_1}.
        \]
    \end{proposition}

    \begin{proof}
        By Proposition \ref{prop:gaussians_bundle_isomorphisms}, $\varphi_\# \eta_i$ is a Gaussian on the product bundle $\cM \times \R^d$; specifically,
        \[
        \varphi_\# \eta_i = N_{\cM \times \R^d}(m_i, \varphi_{m_i} \Sigma \varphi_{m_i}^{-1}).
        \]
        By definition of $d_{\varphi}$, $\varphi$ is an isometry of metric spaces, so induces an isometry at the level of Wasserstein spaces through the pushforward map; that is, 
        \[
        W_2^\varphi(\eta_0,\eta_1) = W_2^{\cM \times \R^d}(\varphi_\# \eta_0, \varphi_\# \eta_1).
        \]
        Applying Lemma \ref{lem:gaussian_distance_product}, we have
        \begin{align*}
        W_2^\varphi(\eta_0,\eta_1) &= W_2^{\cM \times \R^d}(\varphi_\# \eta_0, \varphi_\# \eta_1) \\
        &= d_\cM(m_0,m_1)^2 \\
        &\qquad \qquad + \mathrm{tr}\big(\varphi_{m_0} \Sigma_0 \varphi_{m_0}^{-1})\\
        &\qquad \qquad+ \mathrm{tr}(\varphi_{m_1} \Sigma_1 \varphi_{m_1}^{-1}  - 2 \big((\varphi_{m_0} \Sigma_0 \varphi_{m_0}^{-1})^{\frac{1}{2}} \varphi_{m_1}\Sigma_1\varphi_{m_1}^{-1} (\varphi_{m_0}\Sigma_0 \varphi_{m_0}^{-1})^{\frac{1}{2}}\big)^\frac{1}{2} \big) \\
        &= d_\cM(m_0,m_1)^2  + \mathrm{tr}\big( \Sigma_0   +  \Sigma_1  - 2 \big(\Sigma_0 ^{\frac{1}{2}} \Phi_{m_0,m_1}^{-1}\Sigma_1\Phi_{m_0,m_1} \Sigma_0^{\frac{1}{2}}\big)^\frac{1}{2} \big),
        \end{align*}
        where the last equality follows by linearity and cyclic permutation invariance of trace, together with the fact that the square root of a symmetric positive definite matrix is invariant under change of basis, in the sense that $(\varphi \Sigma \varphi^{-1})^{\frac{1}{2}} = \varphi \Sigma^{\frac{1}{2}} \varphi^{-1}$. 
    \end{proof}

    \paragraph{Distance Between Gaussian Mixtures on Trivial Bundles} Finally, we extend the results above to the setting of Gaussian mixture measures on trivial bundles. 

    We begin with the product bundle $\cM \times \R^d$, with metric $d_{\cM \times \R^d}$ and associated Wasserstein metric $W_2^{\cM \times \R^d}$ defined as above. Let $\mu_i = \sum_{k=1}^{K_i} w_{i}^k \eta_{i}^k$, $i \in \{0,1\}$ be elements of $\GM(\cM \times \R^d)$. As in the classical setting of Gaussian mixtures on $\R^d$, we can consider alternative metrics on the space of Gaussian mixtures. First, let $W_2^{\gauss(\cM \times \R^d)}(\mu_0^\ast, \mu_1^\ast)$ denote the Wasserstein distances between the measures when they are considered as discrete distributions on the space of Gaussians $\gauss(\cM \times \R^d)$; as in the Euclidean setting, we write \[\mu_i^\ast = \sum_{k=1}^{K_i} w_{i}^k \delta_{\eta_{i}^k} \in \mathcal{P}_2(\gauss(\cM \times \R^d)).\] Second, consider the \emph{mixture Wasserstein distance}
    \begin{align*}
    &MW_2^{\cM \times \R^d}(\mu_0,\mu_1)^2 \\
    &\qquad  := \inf_{\pi \in \Pi(\mu_0,\mu_1) \cap \GM_\sigma(\cM \times \R^d)} \int_{(\cM \times \R^d) \times (\cM \times \R^d)} d_{\cM \times \R^d}((p_0,v_0),(p_1,v_1))^2 d\pi ((p_0,v_0),(p_1,v_1)),
    \end{align*}
    where
    \[
    \GM_\sigma(\cM \times \R^d) := \{\pi = \sigma_\# \overline{\pi} \mid \overline{\pi} \in \GM((\cM \times \cM) \times (\R^d \times \R^d))\} \subset \mathcal{P}_2((\cM \times \R^d) \times (\cM \times \R^d)),
    \]
    with $\sigma$ denoting the coordinate permutation map from Definition \ref{def:coordinate_permutation_map}.

    We have the following generalization of Theorem \ref{thm:Gaussian_mixture_metrics} (Theorem \ref{thm:Gaussian_mixture_metrics} is recovered by setting $\cM = \R^d$ and $d_\cM$ to be Euclidean distance), whose proof follows similar ideas as the proof in \cite{delon-agnes:2019}.

    \begin{lemma}\label{lem:gaussian_mixture_vector_bundles}
        Let $\mu_0, \mu_1 \in \GM(\cM \times \R^d)$, with $\mu_i = \sum_{k=1}^{K_i} w_{i}^k \eta_{i}^k$ and $\eta_{i}^k = N_{\cM \times \R^d}(m_{i}^k,\Sigma_{i}^k)$. Then \[MW_2^{\cM \times \R^d}(\mu_0,\mu_1) = W^{\gauss(\cM \times \R^d)}_2(\mu_0^\ast,\mu_1^\ast).\]
    \end{lemma}

    \begin{proof}
        For notational convenience, let $\cE = \cM \times \R^d$ and $\mathcal{F} = (\cM \times \cM) \times (\R^d \times \R^d)$ for the rest of the proof. Let $\omega \in \Pi(\mu_0^\ast,\mu_1^\ast)$ be an optimal coupling for $W_2^{\gauss(\cE)}(\mu_0^\ast, \mu_1^\ast)$, so that 
        \[
        W_2^{\gauss(\cE)}(\mu_0^\ast, \mu_1^\ast)^2 = \sum_{k,\ell} \omega_{k\ell} W_2^{\cE}(\eta_{0}^k,\eta_{1}^\ell)^2.
        \] 
        For each pair of indices $(k,\ell)$, choose an optimal coupling $\pi^{k \ell} \in \Pi(\eta_{0}^{k},\eta_{1}^{\ell})$ for the Wasserstein distance $W_2^\cE(\eta_{0}^{k},\eta_{1}^{\ell})$, and set $\pi = \sum_{k,\ell} \omega_{k\ell} \pi^{k \ell}$ (the superscript notation $\pi^{k \ell}$ is intended to distinguish from the notation $\pi_{k \ell}$, which we used for an entry in a matrix). Then $\pi \in \Pi(\mu_0,\mu_1)$. We also claim that $\pi \in \GM_\sigma(\cE)$. Indeed, from Lemma \ref{lem:gaussian_distance_product}, each $\pi^{k\ell}$ is of the form $\sigma_\# \overline{\pi}^{k\ell}$ for some $\overline{\pi}^{k\ell} \in \gauss(\mathcal{F})$ and, therefore, $\pi = \sigma_\# \overline{\pi}$, where $\overline{\pi} := \sum_{k,\ell} \omega_{k\ell} \overline{\pi}^{k\ell} \in \GM(\mathcal{F})$. Moreover,
        \begin{align*}
            &\int_{\cE \times \cE} d_{\cE}((p_0,v_0),(p_1,v_1))^2 d\pi ((p_0,v_0),(p_1,v_1))\\
            &= \sum_{k,\ell} \omega_{k\ell} \int_{\cE \times \cE} d_{\cE}((p_0,v_0),(p_1,v_1))^2 d\pi^{k\ell}((p_0,v_0),(p_1,v_1))\\
            &= \sum_{k,\ell} \omega_{k\ell} W_2^\cE(\eta_{0}^{k},\eta_{1}^{\ell})^2 = W_2^{\gauss(\cE)}(\mu_0^\ast, \mu_1^\ast)^2.
        \end{align*}
        Since $\pi$ is not necessarily optimal for the mixture Wasserstein distance, it follows that \newline $MW_2^{\cE}(\mu_0,\mu_1) \leq W_2^{\gauss(\cE)}(\mu_0^\ast, \mu_1^\ast)$.

        To see the reverse inequality, suppose that $\pi$ is an arbitrary coupling in the feasible set for $MW_2^{\cE}(\mu_0,\mu_1)$. Let $\pi = \sigma_\# \overline{\pi}$ for some $\overline{\pi} \in \GM(\mathcal{F})$ of the form
        $
        \overline{\pi} = \sum_{j=1}^K \omega_{j} \overline{\pi}^{j},
        $
        where the $\omega_{j}$ are positive real numbers satisfying $\sum_{j}\omega_{j} = 1$ and where each $\overline{\pi}^{j}$ is a Gaussian on the vector bundle $\mathcal{F}$. Therefore $\pi = \sum_{j=1}^K \omega_{j} \pi^{j}$, where $\pi^j := \sigma_\# \overline{\pi}^j$. Let $\rho_i:\cE \times \cE \to \cE$ be coordinate projection onto the left, for $i=0$, or right, for $i=1$, component. The marginal condition on $\pi$ implies that
        \[
        \sum_{k=1}^{K_0} w_0^k \eta_0^k = \mu_0 = (\rho_0)_\# \pi = \sum_{j=1}^K \omega_j (\rho_0)_\# \pi^j.
        \]
        Corollary \ref{cor:identifiability_vector_bundles} then tells us that, for each $j$, there is some $k$ such that $(\rho_0)_\# \pi^j = \eta_0^k$. Moreover, it must be that $(\rho_1)_\# \pi^j = \eta_1^\ell$ for some $\ell$, by the second marginal condition on $\pi$. Putting these observations together and reindexing with double indices, for convenience, we write $\pi = \sum_{k,\ell} \omega_{k\ell} \pi^{k\ell}$, where the marginals of $\pi^{k\ell}$ are $\eta_0^k$ and $\eta_1^\ell$, respectively. It follows that $\omega := (\omega_{k\ell})_{k,\ell}$ defines a coupling of $\mu_0^\ast$ and $\mu_1^\ast$, so a computation very similar to the one in the previous paragraph shows that the value of the mixture Wasserstein distance objective on the coupling $\pi$ is bounded below by $W_2^{\gauss(\cE)}(\mu_0^\ast,\mu_1^\ast)$. Since $\pi$ was an arbitrary element of $\Pi(\mu_0,\mu_1) \cap \GM_\sigma(\cE)$, this completes the proof.
    \end{proof}

    Now consider an arbitrary trivial vector bundle $\cE \to \cM$. For a trivialization $\varphi:\cE \to \cM \times \R^d$, define the \emph{associated mixture Wasserstein distance} between Gaussian mixtures $\mu_0,\mu_1 \in \GM(\cE)$ as 
    \[
    MW_2^{\varphi}(\mu_0,\mu_1)^2 := \inf_{\pi \in \Pi(\mu_0,\mu_1) \cap \GM_\sigma^\varphi(\cE)} \int_{\cE \times \cE} d_{\cE}(w_0,w_1)^2 d\pi (w_0,w_1),
    \]
    where
    \[
    \GM_\sigma^\varphi(\cE) := \{(\varphi^{-1} \circ \sigma)_\# \overline{\pi} \mid \overline{\pi} \in \GM((\cM \times \cM) \times (\R^d \times \R^d))\}.
    \]

    \begin{theorem}\label{thm:main_result}
        With the notation above, $MW_2^{\varphi}(\mu_0,\mu_1) = W_2^{\gauss(\cE)}(\mu_0^\ast,\mu_1^\ast)$, where $\gauss(\cE)$ is considered as a metric space with (the restriction of) Wasserstein distance $W_2^\varphi$.
    \end{theorem}

    \begin{proof}
        We claim that 
        \[
        MW_2^{\varphi}(\mu_0,\mu_1) = MW_2^{\cM \times \R^d}(\varphi_\# \mu_0, \varphi_\# \mu_1 ).
        \]
        Indeed, for any $\overline{\pi} \in \GM((\cM \times \cM) \times (\R^d \times \R^d))$, the equality
        \begin{align*}
        &\int_{(\cM \times \R^d)\times (\cM \times \R^d)} d_{\cM \times \R^d}((p_0,v_0),(p_1,v_1))^2 d(\sigma_\# \overline{\pi})((p_0,v_0),(p_1,v_1)) \\
        & \qquad \qquad \qquad \qquad = \int_{\cE \times \cE} d_\varphi (w_0, w_1)^2 d ((\varphi^{-1} \circ \sigma)_\# \overline{\pi}) (w_0,w_1),
        \end{align*}
        holds, by the change-of-variables formula, the definition of $d_\varphi$ and functoriality of pushforwards. Minimizing over couplings (of $\varphi_\# \mu_0$ and $\varphi_\# \mu_1$) of the form $\sigma_\# \overline{\pi}$ on the left hand side recovers $MW_2^{\cM \times \R^d}(\mu_0,\mu_1)$, and this is equivalent to minimizing over couplings (of $\mu_0$ and $\mu_1$) of the form $(\varphi^{-1} \circ \sigma)_\# \overline{\pi}$ on the right hand side, which recovers $MW_2^\varphi(\mu_0,\mu_1)$. 

        Next, we show that 
        \[
        W_2^{\gauss(\cE)}(\mu_0^\ast,\mu_1^\ast) = W_2^{\gauss(\cM \times \R^d)}((\varphi_\# \mu_0)^\ast, (\varphi_\# \mu_1)^\ast).
        \]
        Observe that
        \[
        \varphi_\# \mu_i = \sum_k w_{ik} \varphi_\# \eta_{ik} \Rightarrow (\varphi_\# \mu_i)^\ast = \sum_k w_{ik} \delta_{\varphi_\# \eta_{ik}}.
        \]
        It follows that $\Pi(\mu_0^\ast,\mu_1^\ast) = \Pi((\varphi_\# \mu_0)^\ast, (\varphi_\# \mu_1)^\ast)$ (considered as elements of $\R^{K_1 \times K_2}$). For any coupling $\omega$, the fact that $\varphi_\#$ is an isometry from $W^\cE_2$ to $W^{\cM \times \R^d}_2$ implies
        \[
        \sum_{k,\ell} \omega_{k\ell} W_2^\cE(\eta_{0k},\eta_{1\ell})^2 = \sum_{k,\ell} \omega_{k\ell} W_2^{\cM \times \R^d}(\varphi_\# \eta_{0k},\varphi_\# \eta_{1\ell})^2,
        \]
        and this proves the claim.
        
        Now the result follows from Lemma \ref{lem:gaussian_mixture_vector_bundles}:
        \[
        W_2^{\gauss(\cE)}(\mu_0^\ast,\mu_1^\ast) = W_2^{\gauss(\cM \times \R^d)}((\varphi_\# \mu_0)^\ast, (\varphi_\# \mu_1)^\ast) = MW_2^{\cM \times \R^d}(\varphi_\# \mu_0, \varphi_\# \mu_1) = MW_2^\varphi(\mu_0,\mu_1).
        \]
    \end{proof}

\paragraph{Summary of the Mixture Wasserstein Distance} For the convenience of the reader, we summarize the notation and constructions of the previous subsections here. The mixture Wasserstein distance for a trivial vector bundle $\cE \to \cM$ with trivialization $\varphi:\cE \to \cM \times \R^d$, $MW_2^\varphi$, between Gaussian mixtures
\[
\mu_i = \sum_{k=1}^{K_i} w_i^k \eta_i^k \in \GM(\cE), \quad \eta_i^k = N_\cE(m_i^k, \Sigma_i^k), \quad i \in \{0,1\}
\]
can be expressed as 
\begin{equation}\label{eq: wasserstein-type distance}
MW_2^\varphi(\mu_0,\mu_1)^2 = \min_{\omega \in \Pi(\mu_0,\mu_1)} \sum_{k,\ell} \omega_{k\ell} W_2^\cE(\eta_{0}^{k},\eta_{1}^{\ell})^2
\end{equation}
where
\[
W_2^\cE(N_\cE(m_0, \Sigma_0),N_\cE(m_1, \Sigma_1))^2 =\sum_{k,\ell}
    d_\cM(m_0,m_1)^2  + \mathrm{tr}\big( \Sigma_0   +  \Sigma_1  - 2 \big(\Sigma_0 ^{\frac{1}{2}} \Phi_{m_0,m_1}^{-1}\Sigma_1\Phi_{m_0,m_1} \Sigma_0^{\frac{1}{2}}\big)^\frac{1}{2} \big)
\]
and
\[
\Phi_{m_0,m_1} = \varphi_{m_1}^{-1} \circ \varphi_{m_0}: \cE_{m_0} \to \cE_{m_1}.
\]

\subsection{Choosing Trivializations}\label{sub: trivialization selection}

In this subsection, we address the question of how to choose an appropriate trivialization to fit into the pipeline described in the preceding subsection. We focus on the case that the vector bundle is of the form $T\cM \to \cM$, where $\cM$ is a parallelizable Riemannian manifold. We will make the simplifying assumption that the manifold is endowed with a distinguished point $p \in \cM$ and a distinguished isometry $P_{m}:T_p \cM \to T_m \cM$ for every point $m \in \cM$, such that $P_p$ is the identity map and such that the isometries vary smoothly in $m$. 

\begin{example}\label{ex:manifold_data}
    We now provide some examples of where the data described above might arise. 
    \begin{enumerate}
        \item In the numerical experiments in the following section, the underlying manifold $\cM$ is a \emph{punctured sphere}---that is, a sphere $\s^n$ with a point removed. In this case, the distinguished point $p$ is taken to be antipodal to the removed point. For every $m \in \cM$, there is a unique geodesic joining $p$ to $m$, and the isometry $P_m:T_p \cM \to T_m \cM$ is given by parallel transporting along the geodesic. 

        In practice, one begins with a Gaussian mixture on the sphere $\s^n$ and chooses a distinguished point $p \in \s^n$, say, by always taking the north pole or by choosing a Fr\'{e}chet mean of the basepoints of the Gaussians. Generically, the antipodal point $-p$ will not be a basepoint of any Gaussian in the mixture, so that we can restrict our analysis to the setting of the punctured sphere $\cM = \s^n \setminus \{-p\}$.

        \item More generally, if there is some point $p \in \cM$ such that there is a unique geodesic to any other point $m \in \cM$, then the parallel transport approach described above defines the desired collection of isometries. This is the setting used in \cite{WGOT} for comparing wrapped Gaussians.

        \item If $\cM$ is a Lie group with left (or right)-invariant Riemannian metric, then the natural choice of basepoint is the identity $e \in \cM$, and the isometry $P_m$ is the derivative of the left (or right) multiplication map.
    \end{enumerate}
\end{example}

Given $p \in \cM$ and $\{P_m\}_{m \in \cM}$ as above, we define a bundle isometry $T\cM \to \cM \times \R^d$ as follows. Choosing any orthonormal basis $F = (f_1,\ldots,f_d)$ for $T_{p} \cM$, one can extend this to a global moving frame for $\cM_p$ by defining $F(m) = \{f_j(m) = P_m(f_j), j = 1,...,d \}$. Let $\varphi^{p,F}:T\cM \to \cM \times \R^d$ be defined by 
 \[
 \varphi^{p,F}(v) = \left(m, \big(\langle v,f_j(m), \rangle_m\big)_{j=1}^d \right) \qquad \mbox{for $v \in T_m\cM$.}
 \]

We now observe that, while this construction depends on choices of $p$, $\{P_m\}_{m \in \cM}$ and $F$, the induced distance $MW_2^{\varphi^{p,F}}$  is actually invariant under the choice of $F$. Since $p$ and $\{P_m\}_{m \in \cM}$ arise naturally in several examples of interest (Example \ref{ex:manifold_data}), the following result shows that the mixture Wasserstein distance is a natural tool in practice.

\begin{proposition}\label{prop: basis invariance}
    With the notation defined above, the mixture Wasserstein distance \newline $MW_2^{\varphi^{p,F}}$ does not depend on the choice of orthonormal basis $F$. That is, for any orthonormal bases $F=(f_1,\ldots,f_d)$ and $G = (g_1,\ldots,g_d)$ for $T_p \cM$, we have $MW_2^{\varphi^{p,F}}(\mu_0,\mu_1) = MW_2^{\varphi_{p,G}}(\mu_0,\mu_1)$ for any Gaussian mixtures $\mu_0,\mu_1$.
\end{proposition}

\begin{proof}
        By Theorem \ref{thm:main_result}, 
        $MW_2^{\varphi^{p,F}}(\mu_0,\mu_1) = W_2^{\gauss(\cE)}(\mu_0^\ast,\mu_1^\ast)$, where, for $\mu_i = \sum_{k=1}^{K_i} w_i^k \eta_i^k$,
        \[
        W_2^{\gauss(\cE)}(\mu_0^*, \mu_1^*)^2 = \min_{\pi \in \Pi(w_0, w_1)} \Sigma_{k,l} \pi_{kl} W_2^{\varphi^{p,F}}(\eta_{0}^k, \eta_{1}^\ell)^2.
        \]
        By Proposition \ref{prop:trivial_bundle_gaussian_distance}, for $\eta_j = N_\cE(m_j,\Sigma_j)$, 
        \[
        W_2^{\varphi^{p,F}}(\eta_0,\eta_1) = d_\cM(m_0,m_1)^2 + \mathrm{tr}\big(\Sigma_0 + \Sigma_1 - 2 \big(\Sigma_0^{\frac{1}{2}} (\Phi_{m_0,m_1}^{p,F})^{-1} \Sigma_1 \Phi_{m_0,m_1}^{p,F} \Sigma_0^{\frac{1}{2}}\big)^\frac{1}{2} \big),
        \]
        with $\Phi_{m_0,m_1}^{p,F} = (\varphi_{m_1}^{p,F})^{-1} \circ \varphi_{m_0}^{p,F} : \cE_{m_0} \to \cE_{m_1}$. We claim that 
        \[(
        \varphi_{m_1}^{p,F})^{-1} \circ \varphi_{m_0}^{p,F} = P_{m_1} \circ P_{m_0}^{-1};
        \]
        since the map on the right hand side does not depend on $F$, the result follows. The claim is verified by a direct calculation: for $v \in T_{m_0} \cM$, write $v = \sum_j v_j f_j(m_0)$. Then
        \[
        P_{m_1} \circ P_{m_0}^{-1}(v) = P_{m_1}\big(\sum_j v_j f_j\big) = \sum_j v_j f_j(m_1)
        \]
        and
        \[
        (\varphi_{m_1}^{p,F})^{-1} \circ \varphi_{m_0}^{p,F}(v) = (\varphi_{m_1}^{p,F})^{-1}\big((v_j)_j\big) = \sum_j v_j f_j(m_1).
        \]
\end{proof}

\begin{remark}
    The result and its proof show that, in the situation described in this subsection, the part of $W_2^{\gauss(\cE)}$ that involves comparing Gaussians on $\cE$ can be computed by moving each Gaussian to the basepoint; that is, $\eta = N_\cE(m,\Sigma)$ can be replaced with $(P_m^{-1})_\# \eta$. Independence of choice of frame in constructing the trivialization then just amounts to the fact that Wasserstein distance on a Euclidean space is invariant under isometries.
\end{remark}

\section{Experimental Studies}\label{section: applications}

In this section, we demonstrate applications of Wasserstein-type distances on some nonlinear domains. We first introduce a general setup and present procedures for estimating Gaussian mixture parameters from sample data. Then we move to examples involving simulated data on a unit sphere and the Kendall shape space of planar triangles. Finally, we present an application on real data in the preshape space of planar closed shapes. We note that the Python package \textbf{geomstats} (\cite{geomstats}) was used for geometric computations in these experiments.

\subsection{Basic Experimental Setup}\label{sub: general case}

Let $\cM$ be a d-dimensional parallelizable Riemmanian manifold. One can use the following procedure to define a global moving frame on $\cM$: 1) choose a distinguished point $p \in \cM$, 2) randomly generate $d$ linearly independent tangent vectors in $T_p\cM$; 3) calculate a tangent space PCA of those points, and take the tangent space principal components as an orthonormal basis $F$. Given a reference point $p$ and a reference frame $F$, we calculate an element of the global moving frame $F(m)$ by parallel transporting the tangent vectors in $F$ to $T_m(\cM)$ along the unique geodesic from $p$ to $m$ in $\cM$.

The next issue is the estimation of Gaussian mixture parameters from given samples on $\cM$. As mentioned earlier, there exist computational solutions for estimating mean, covariance, and weight parameters (see for example \cite{Hauberg2018DirectionalSW}) for mixture components but they can become costly especially on nonlinear domains. Computing a Gaussian mean on manifolds is already an iterative procedure, and when combined with an outer loop over mixture components it becomes prohibitive. Instead of expectation maximization-type solutions, we make use of clustering-based approaches: Riemmanian $K$-means (\cite{geomstats}) and $K$-Modes Kernel Mixtures (KMKM) clustering (\cite{deng-ICPR:2022}). Riemannian $K$-means is still an expensive method and we use it for low-dimensional examples. For estimation on high-dimensional shape manifolds, we utilize the $K$-mode clustering method.

    Riemmanian $K$-means is an extension of the well known $K$-means algorithm to nonlinear manifolds. Given a data set $X = \{ x_i \in \cM,\  i = 1,...,n\}$, the $K$-means algorithm returns a cluster index $\mathcal{I} = \{\ell_i \in (1,...,K), i =1,...,n\}$ representing the cluster assignments for each $x_i$. In order to obtain estimates for Gaussian mixture parameters, we treat the clusters found by $K$-means as samples from Gaussian mixture components. Let $\hat{m}_k$ denote the Fr\'{e}chet mean \cite{frechet} for the sample points in cluster $k$, and let $n_k$ denote the number of sample points in cluster $k$. We estimate the weight for component $k$ as $\hat{w}_k = n_k/n$, and calculate the tangent space covariance in the coordinates of the moving frame $F(m) = \{ f_j(m), j = 1,...,m\}$ as $\hat{\Sigma}_k = \frac{1}{n_k-1} V^t V$, where $V_{j,i} = \langle \log_{\hat{m}_k}(x_i), f_j(m_k) \rangle$. Our estimate of the Gaussian mixture is then $\hat{\mu} = \Sigma_{k=1}^{K} \hat{w}_k N_{\cE}(\hat{m}_k, \hat{\Sigma}_k )$.

    In the second approach for estimating model parameters, we use a $k$-modes kernel-mixture clustering algorithm presented in \cite{deng-ICPR:2022,deng-ISBI:2022}. This metric-based nonparametric procedure takes in the matrix of pairwise shape metric  between all data points and uses it to cluster individual shapes around estimated modes. Modes are significant local maxima of the underlying probability distribution and are computed as the points with the most neighbors. An important strength of this approach is that neighborhoods are determined automatically from the data, avoiding the need for any manual selection of hyperparameters. The output of this procedure are: (1) cluster label for most of input points, (2) points that are modes of each clusters, and (3) outliers that do not belong to any cluster. We treat each cluster as a mixture component, with its mode as the estimated mean $\hat{m}_k$. Furthermore, we compute a covariance matrix $\hat{\Sigma}_k$ in the tangent space of this estimated mean, with respect to the global frame. These steps are same as in the previous item and result in a Gaussian mixture $\hat{\mu}$.

Given estimates of Gaussian mixtures $\hat{\mu}_{i} = \sum_{k=1}^{K_i} w_k N(m_k, \Sigma_k)$, $i =0,1$, parameterized with respect to a moving frame $(p,F)$, we calculate the Wasserstein-type distance between them by using linear programming to solve:
\begin{equation}\label{eq: mixture distance}
    {MW_2^{\varphi_{p,F}}}(\hat{\mu}_{0},\hat{\mu}_{1})^2 = \underset{w \in \Pi(w_0,w_1)}{\min}  \sum_{k,\ell}^{K_0,K_1} w_{k\ell} W_2^{\cE}(\hat{\mu}_{0}^k,\hat{\mu}_{1}^\ell)^2
\end{equation} 
In the following we take specific examples of $\cM$ and demonstrate some applications of $MW_2$.

\subsection{Simulated data on Punctured 2-Sphere}\label{sub: sphere}

\begin{figure}
    \centering
    Moving Frame at Sample Points\\
    \includegraphics[width=.4\textwidth]{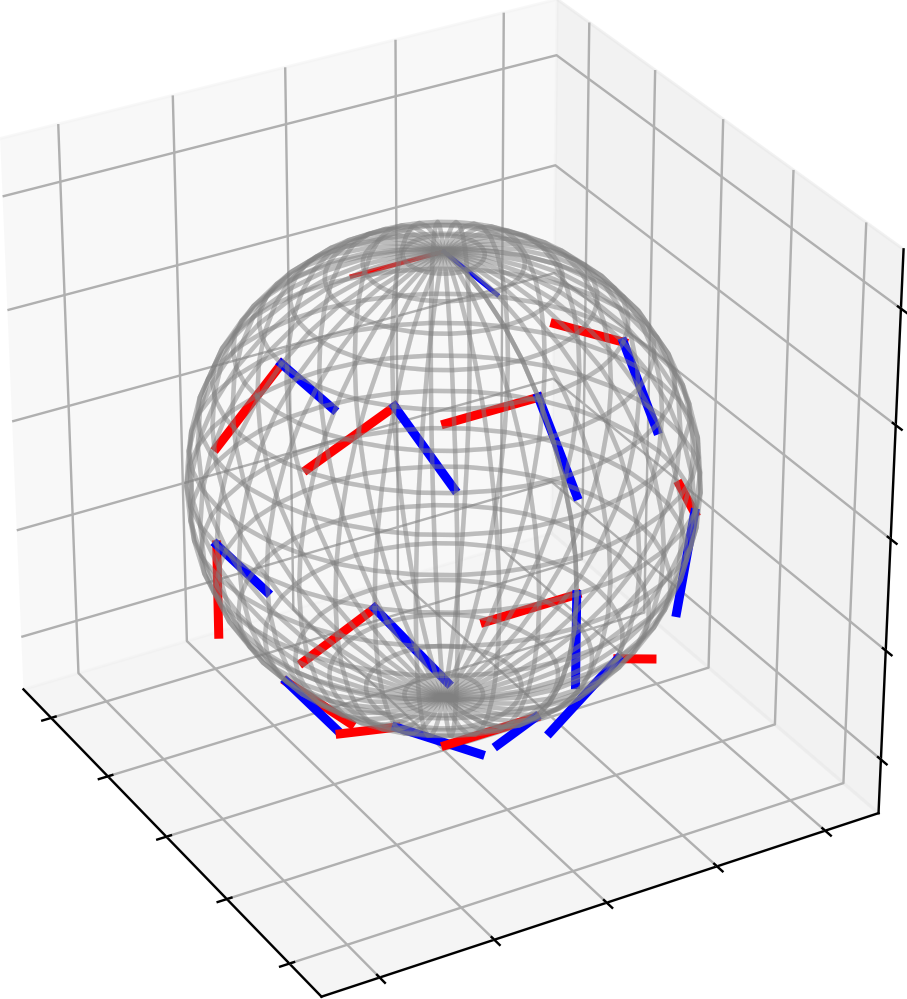}
    \includegraphics[width=.4\textwidth]{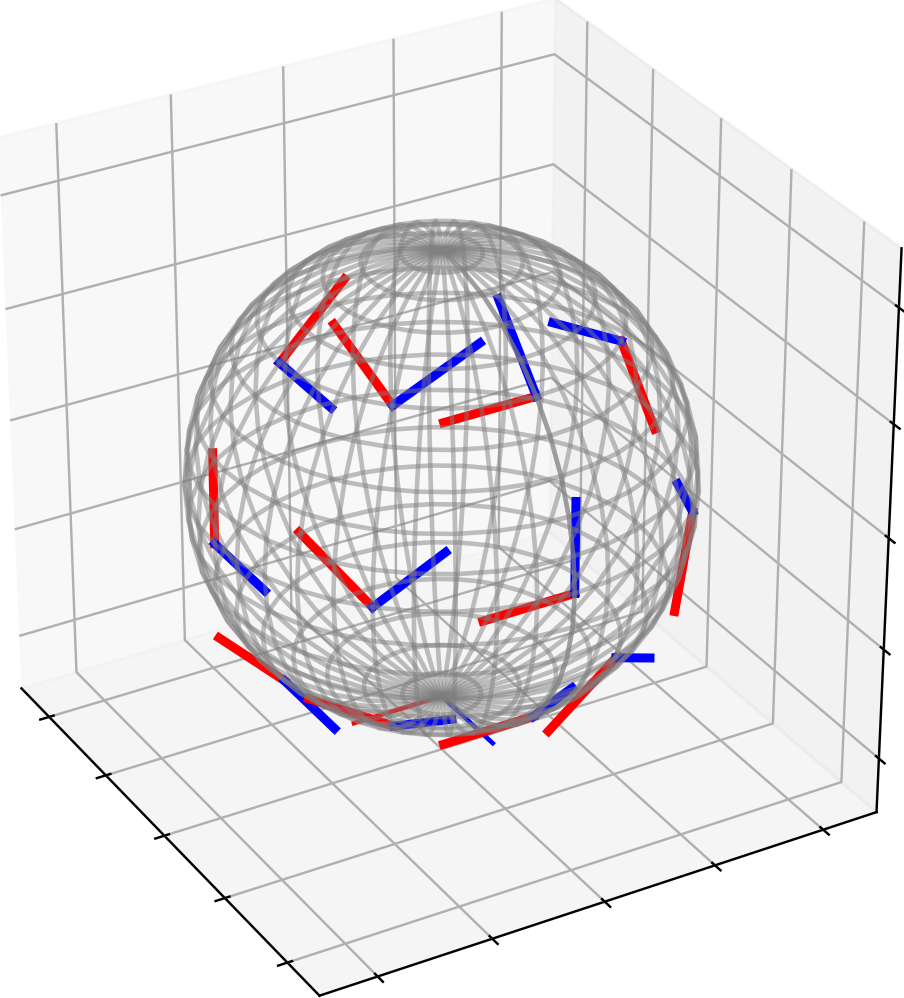}
    \caption{Examples of two reference frames transported to a select number of points on the punctured sphere, providing a consistent coordinate system over the entire manifold. Left: Moving frames, transported from reference point [0,0,1]; Right: Moving frames, transported from reference point [0,0,-1]}
    \label{fig: moving frame}
\end{figure}

As the first example, we consider the domain $\mathbb{S}^2$ and simulate several Gaussian mixtures on $\s^2$ to demonstrate the MW metric. In this experiment, we select either $p = [0,0,1]^T$ or $-p = [0,0,-1]^T$ as a distinguished point and we set $F = [[-1,0,0]^T, [0,-1,0]^T] \in \mathbb{R}^{3\times2}$ as a reference frame---we then assume that Gaussian mixtures have no component in the tangent space of $-p$ (respectively, in the tangent space of $p$). Figure \ref{fig: moving frame} (left) shows the reference frame transported to several points on the $\s^2$. The right side shows a similar plot when the distinguished point is $[0,0,-1]$ instead, illustrating the dependence of the moving frame on basepoint. 

To generate sample data $X$ from a Gaussian with mean $m$ and covariance $\Sigma$ on $\s^2$, we use the following steps. Recall that the moving frame at $m$ is given by $F(m)$.  We generate a set $V = \{v_i \sim N(0, \Sigma), \Sigma \in \mathsf{Sym}_d^+, i = 1,...,N\}$. Then, we compute $X_i = \exp_m(\big(\langle v_i,f_j(m), \rangle_m\big)_{j=1}^d)$ to obtain samples from a Gaussian. Further,  to generate samples from a Gaussian mixture, we just generate samples from individual Gaussians with frequencies proportional to their weights. Given a set of points $X \in \s^2$, represented in terms of their extrinsic coordinates $X \in \mathbb{R}^{N \times 3}$, one can fit a Gaussian mixture model to those points using the Riemmanian K-means approach described above. Figure \ref{fig:fit frame} illustrates this process for a chosen moving frame. The left panel shows samples from a Gaussian mixture with two components, the middle panel shows a data clustering using colors, and the right panel shows an moving frame at $p$ transported to the Fr\'{e}chet means of the clusters. We use these moving frames to express the tangent-space covariance of the data in each cluster, as discussed in Section \ref{sub: general case}. The relative weights of the components $\{w_k\}$ are estimated using the relative frequencies. The resulting Gaussian mixture is denoted by $\hat{\mu} = \Sigma_{k=1}^K \hat{w}_k N(\hat{m}_k, \hat{\Sigma}_k)$.

Given two such estimated Gaussian mixtures $\hat{\mu}_{i}$, $i =0,1$, parameterized with respect to a moving frame $(p,F)$, we calculate the Wasserstein-type distance between them using linear programming to solve Eqn.~\ref{eq: mixture distance}.

 \begin{figure}
    \centering
    Identifying Gaussian Mixture Parameters with Respect to Moving frame\\
    \includegraphics[width=.32\textwidth]{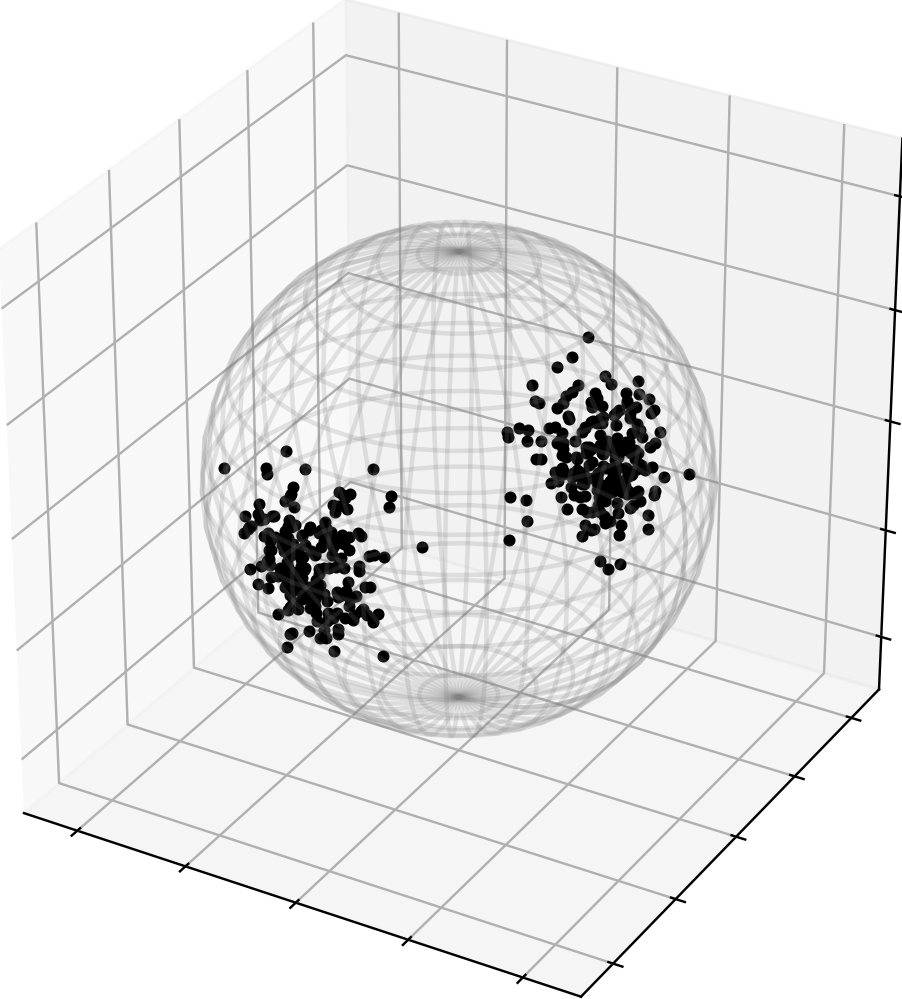}
    \includegraphics[width=.32\textwidth]{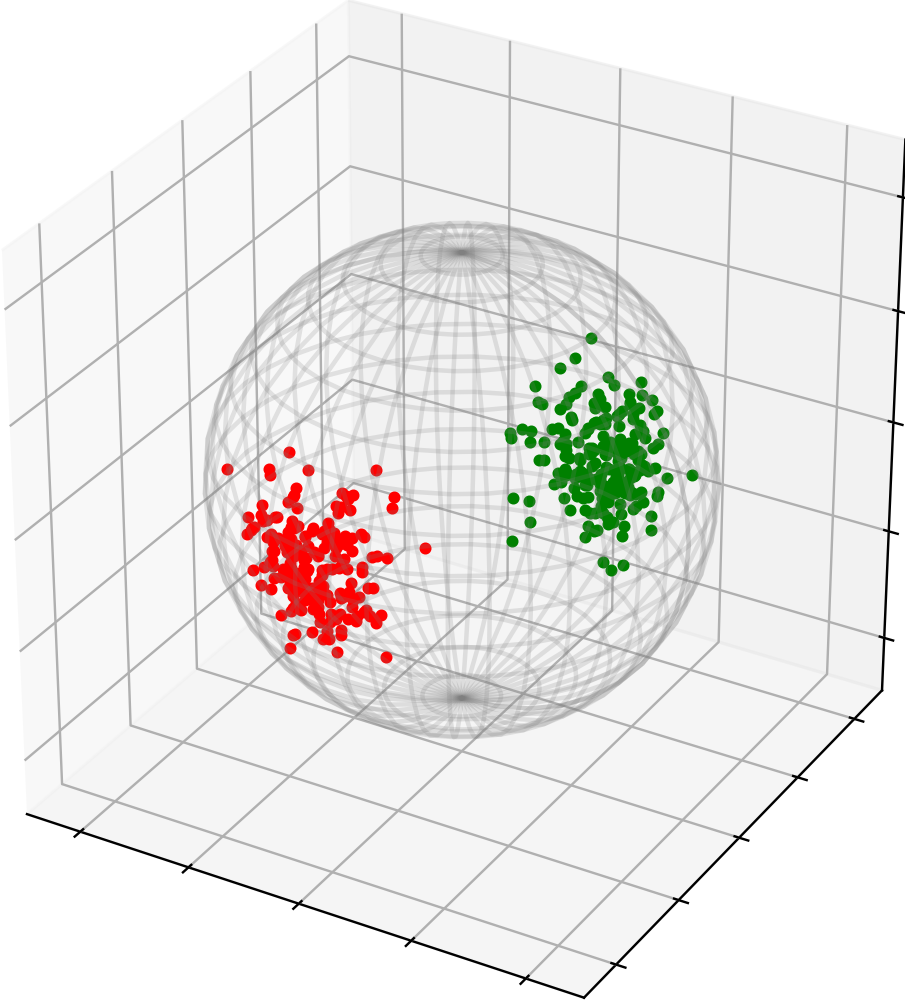}
    \includegraphics[width=.32\textwidth]{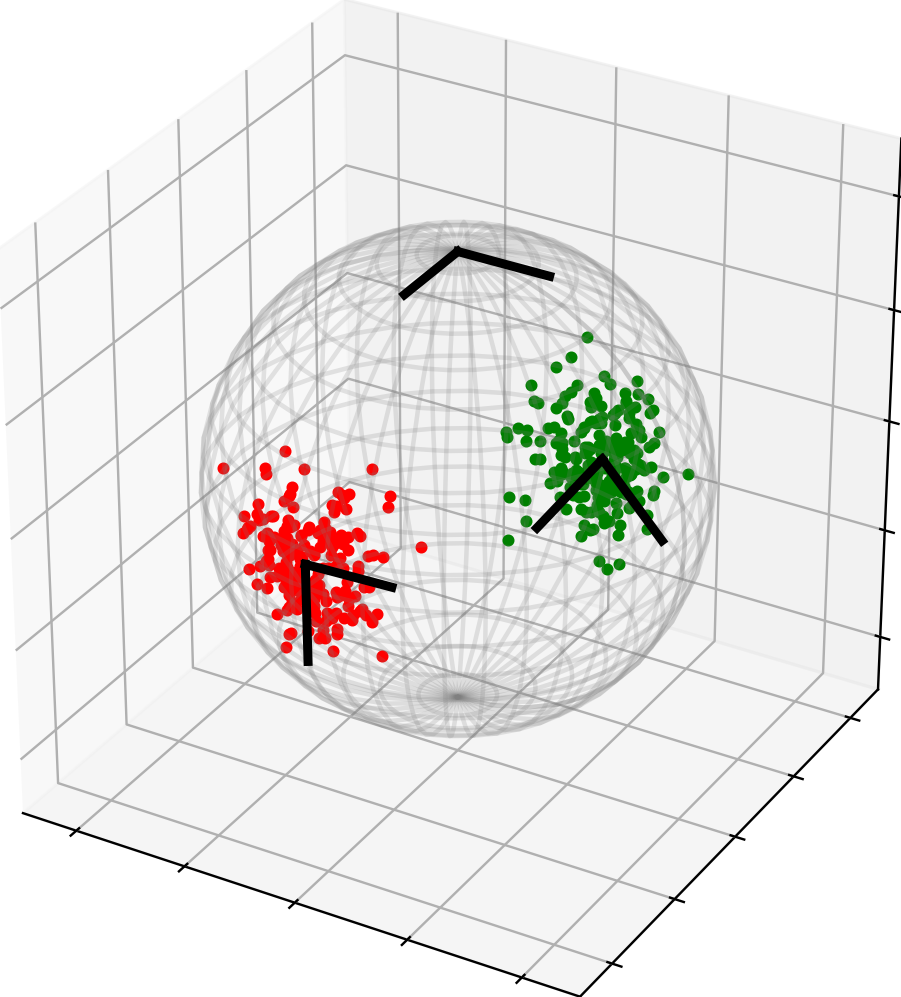}
    \caption{Steps for estimating Gaussian mixture parameters with respect to a moving frame on a punctured $\s^2$. Left: Sphere with sample data; Middle: Sample data colored by K-means cluster assignment; Right: Moving frame transported to Fr\'{e}chet means of clusters.}
    \label{fig:fit frame}
\end{figure}

Figure \ref{fig: experiment2} shows several examples of  evaluations of the Wasserstein distance between Gaussian mixtures estimated from the same sets of data, but with different moving frames. In particular, the two moving frames used the distinguished points $[0,0,1]^T$ and $[0,0,-1]^T$, respectively. The titles of the plots contain the Wasserstein-type distance calculated with respect to the moving frame, and color represents a Gaussian mixture. Example 1 shows a case where the Gaussian mixtures differ significantly, in terms of means, covariances and weights. Consequently, the Wasserstein-type distance between them is large: $0.6098$. Using a different moving frame, we obtain a similar value of $0.6071$ emphasizing the relative stability of this computation with respect to the choice of moving frame. In Example 2, the two Gaussian mixtures are relatively similar: they have the same means and covariances, and differ only in weights. In both cases, we observe that the two moving frames result only in small numerical differences in the distances.

\begin{figure}
    \centering
     \begin{tabular}{|cc|cc|}
     \hline
            \includegraphics[width=.223\textwidth]{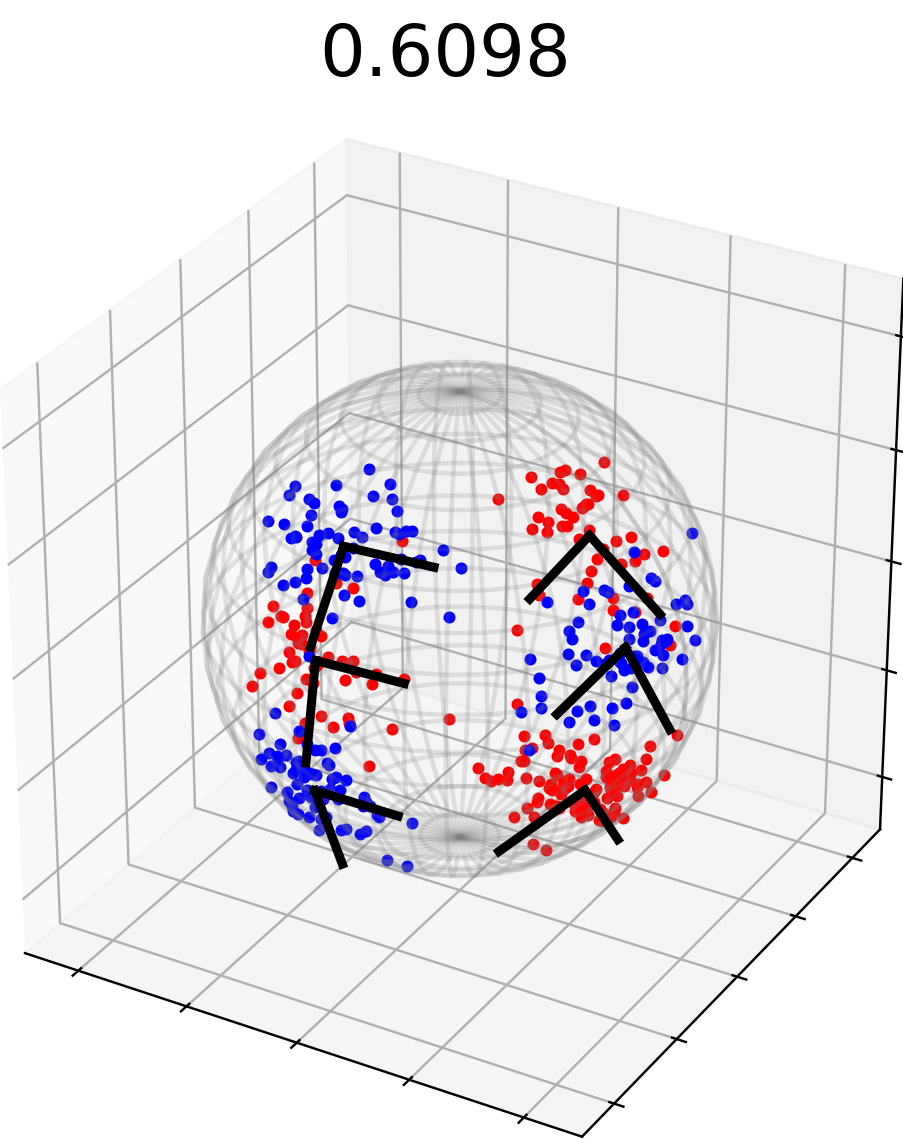} &
            \includegraphics[width=.223\textwidth]{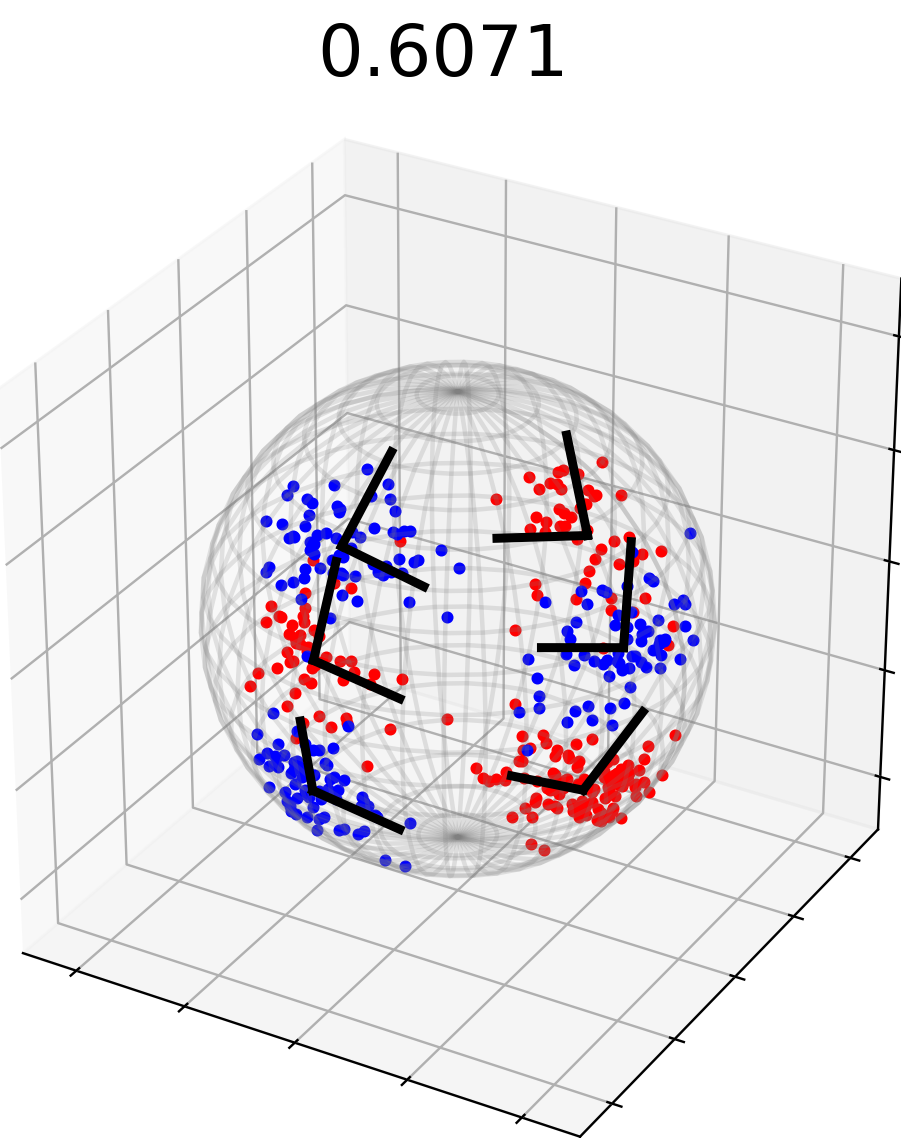} &
            \includegraphics[width=.223\textwidth]{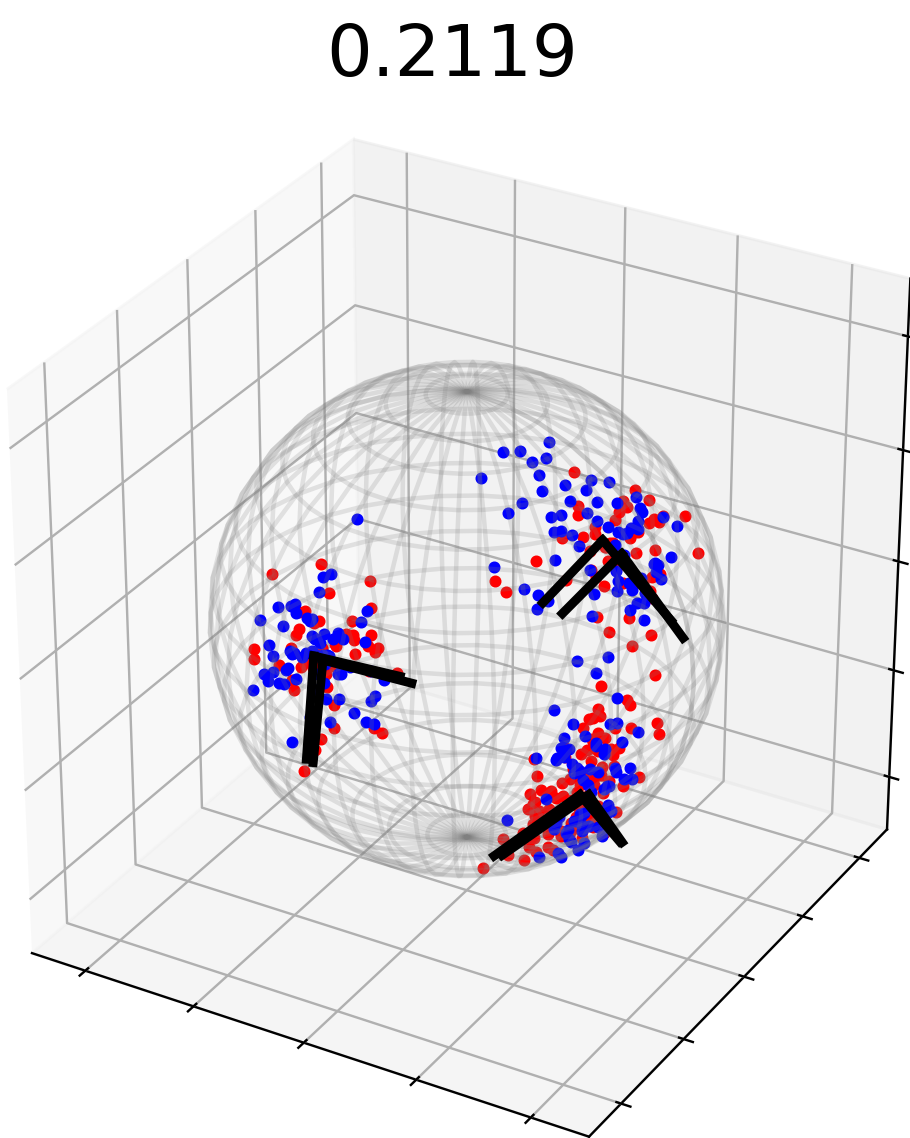} &
            \includegraphics[width=.223\textwidth]{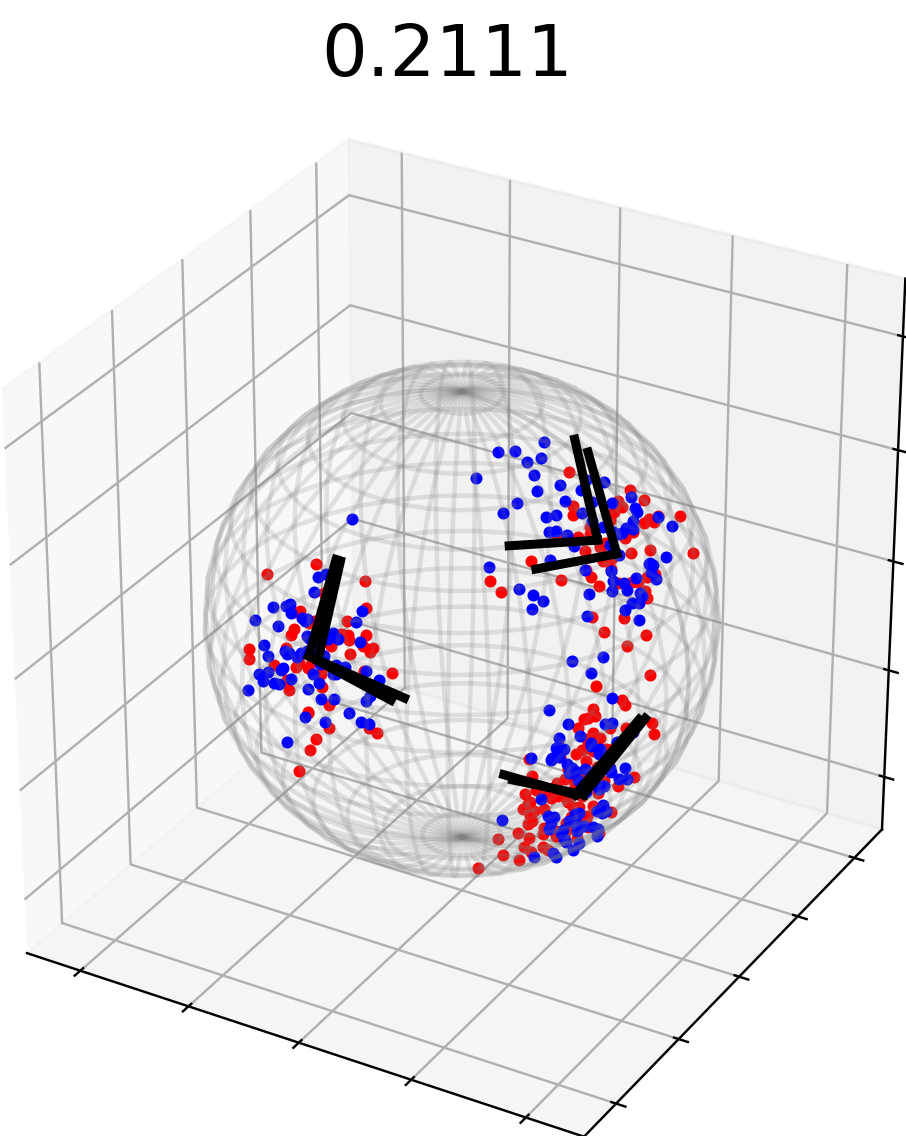} \\
            \multicolumn{2}{|c|}{Example 1} & \multicolumn{2}{c|}{Example 2} \\
            \hline
            \end{tabular}
    \caption{Examples of Wasserstein-type Distances between Gaussian mixtures on the tangent bundle of a punctured $\s^2$ computed using two different moving frames. Color denotes a Gaussian mixture; plot titles show calculated distances. Example 1; Samples from two Gaussian mixtures with different means, covariances, and weights. Example 2: Samples from two Gaussian mixtures with same means and covariances, but different weights.
    }
    \label{fig: experiment2}
\end{figure}

\subsection{Comparing Shape Populations of Triangles}

In this paper we are interested in shape spaces of objects as the domains for imposing and comparing probability distributions. In other words, we want to compare shape populations, modeled as Gaussian mixtures, using Wasserstein-type distances. Recall that shape is a geometric property that is invariant to rotation, translation, and scaling. Before we consider shapes of planar contours, we analyze a simpler case analyzing shapes of planar triangles. The shape space of planar triangles is denoted by the quotient space $\s^3/SO(2)$ which can be further identified with $\s^2$~\cite{Kendall_triangle}. Hence, our analysis of triangle shapes is performed on a (punctured) $\s^2$. 
The steps of the computation---establishing a moving frame, estimating Gaussian mixture parameters, and calculating Wasserstein-type distances between Gaussian mixtures---are similar to the previous subsection. 

Now we provide details for the $\s^2$ representation of planar triangles. 
Let $\{x_i \in \real^2, i = 1,2,3\}$ be the set of all planar traingles. We identify $x_i$ with elements $z_i\in \mathbb{C}$, such that $z_i = (x_{i,1} + j x_{i,2})$, $j = \sqrt{-1}$. After we remove rigid translations and global scaling, we obtain the set
$\mathcal{P}_T = \{ z \in \mathbb{C}^3 | \frac{1}{k}\sum_{i=1}^3 z_i = 0, \|z\| = 1 \}$. 
$\mathcal{P}_T$ is referred to as the {\it preshape space}, because we have not yet removed rigid rotations. The shape space of 2D triangles is thus:
${\cal S}_T = \{ [z] = \{e^{j\varphi}z| \varphi \in \mathbb{S}^1, z \in \mathcal{P}_T \}\}$.
An element $[z] \in {\cal S}_T$ corresponds to a unique triangular shape, with the parameter $\varphi$ denoting its  rotation with respect to a chosen coordinate system. Any $[z] \in {\cal S}_T$ can be isometrically mapped to a point on $\mathbb{S}^2$ using the Hopf Fibration presented in Appendix~\ref{hopf equations}, and we use this fact to estimate parameters for and calculate Wasserstein-type distances between Gaussian mixtures in ${\cal S}_T$.

\begin{figure}
    \centering
    \includegraphics[width = .99\textwidth]{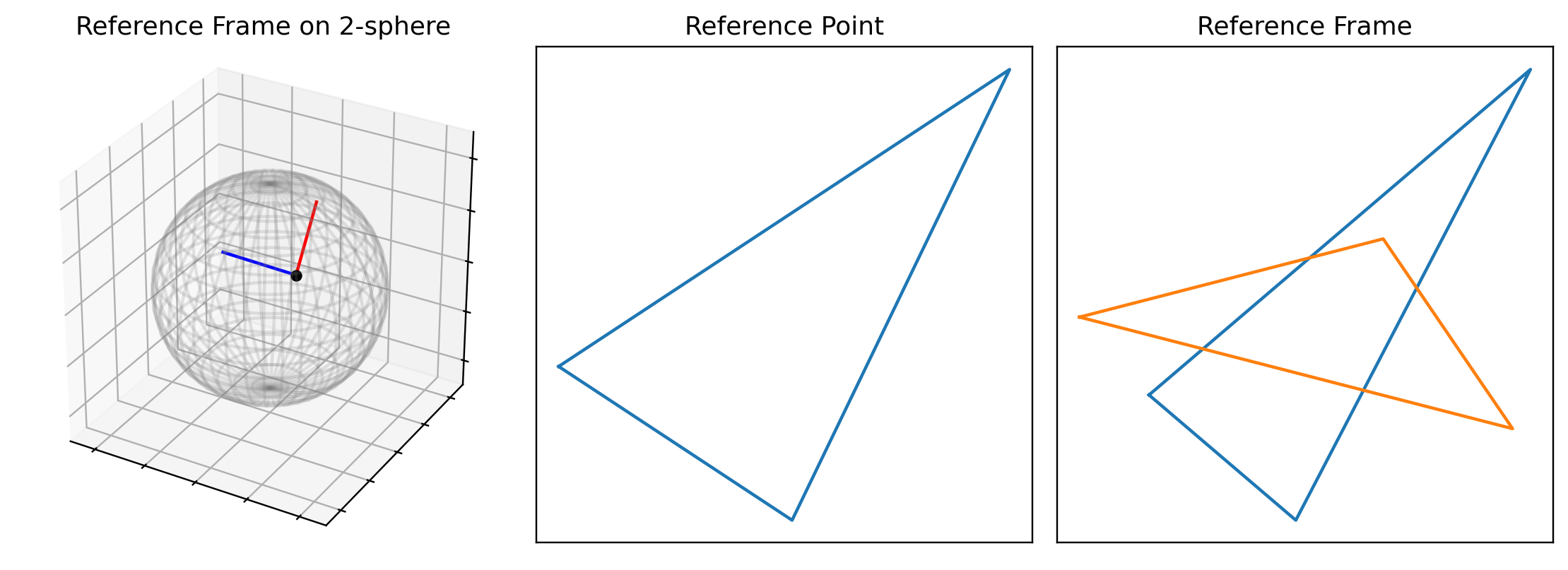}
    \caption{A reference point and a reference frame for Kendall shape space of 2D triangles. Left: Representation of moving frame under Hopf map on $\s^2$; Middle: Triangular representation of reference point; Right: Triangular representations of tangent vectors in reference frame.}
    \label{fig: triangle reference frame}
\end{figure}

Similar to the previous section, we arbitrarily select a reference point $p \in{\cal S}_T$, use the Hopf fibration to map it to $\tilde{p} \in \mathbb{S}^2$, and generate a random basis $F$ for the tangent space $T_{\tilde{p}} \s^2$. One such moving frame is shown in the left panel of Fig.~\ref{fig: triangle reference frame}, along with the triangular representations of the reference point, and the tangent vectors in the reference frame, shown in the middle and left panels of Fig.~\ref{fig: triangle reference frame}, respectively. Next, we generate random samples from Gaussian mixtures on $\mathbb{S}^2$ using the procedure outlined before. Given sample data and a moving frame for $\mathbb{S}^2$, we use the Riemannian K-means algorithm described in Section \ref{sub: general case} to estimate Gaussian mixture parameters from the sample data. Plots of the sample data, colored by cluster assignment, are presented in the first and third plots of Fig.~\ref{fig: triangle sample data}. The accompanying panels show these colored points as planar triangles to visualize clustered shapes. Given parameter estimates, we can calculate the Wasserstein-type distance using Eqn.~\ref{eq: mixture distance}.

\begin{figure}
    \centering
    \includegraphics[width=.99\textwidth]{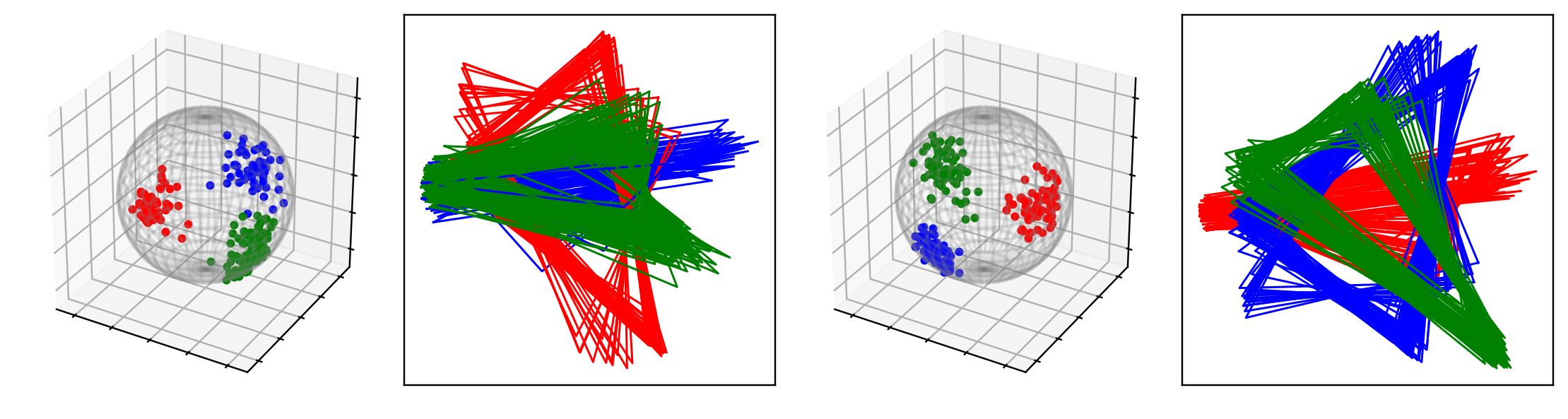}
    \caption{Samples from two Gaussian mixtures defined on the tangent bundle of ${\cal S}_T$, colored by cluster assignment. The triangles in panel 2 (resp. 4) correspond to the points on the sphere in panel 1 (resp. 3).   }
    \label{fig: triangle sample data}
\end{figure}

\begin{figure}
    \centering
    \includegraphics[width=.99\textwidth]{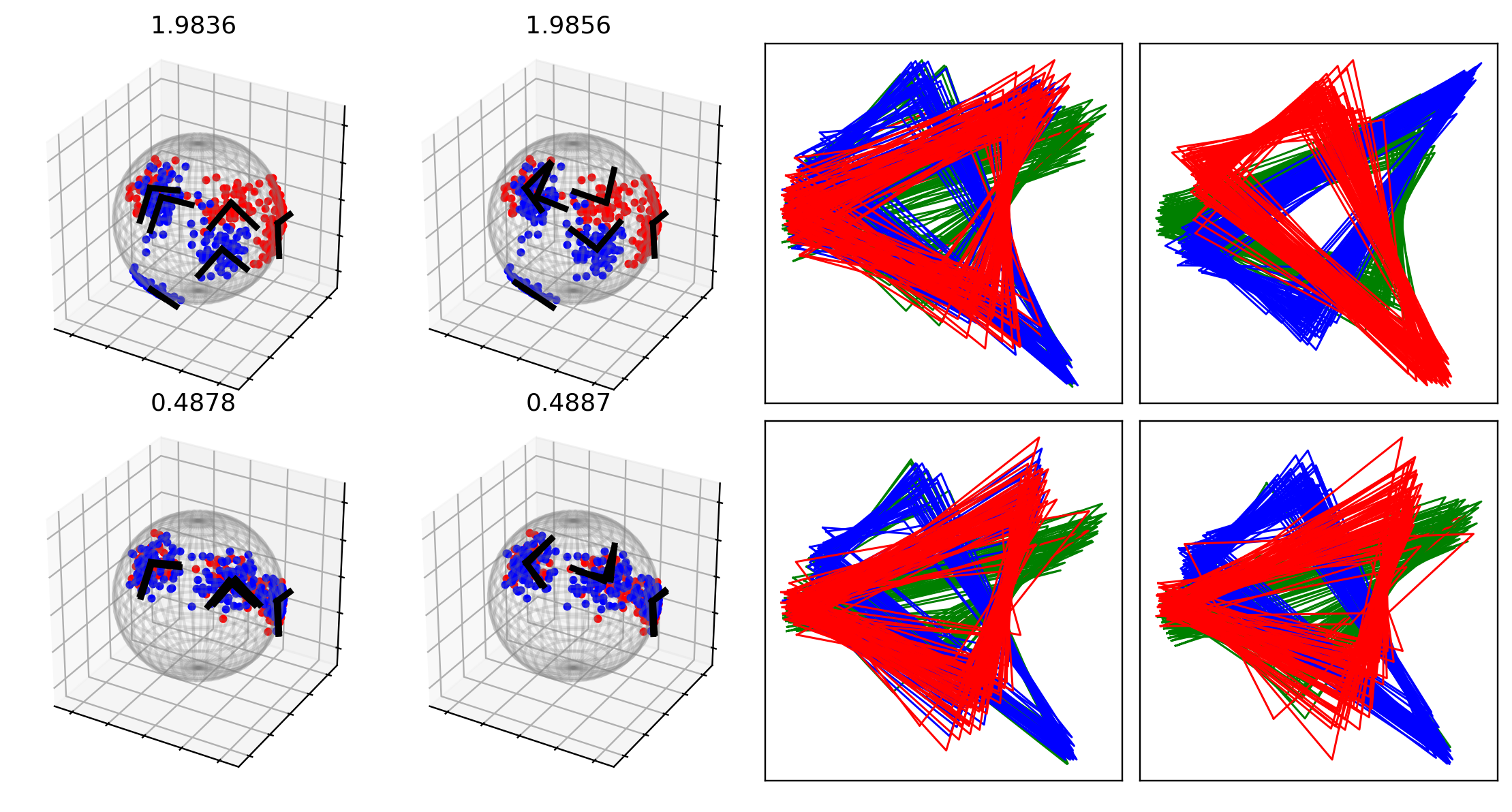}
    \caption{Comparison of Wasserstein-type Distances between Gaussian mixtures on ${\cal S}_T$, estimated with the same procedure from the same data, using two different moving frames. Top Left; Gaussian mixtures with different means, covariances, and weights. Bottom Left; Gaussian mixtures with same means and covariances, but different weights. Top Right: Triangular representations of sample points from Gaussian mixtures for experiment in top row (left corresponds to blue, right corresponds to red); Bottom Right: Triangular representations of sample points from Gaussian mixtures for experiment in bottom row (left corresponds to blue, right corresponds to red) }
    \label{fig: experiment4}
\end{figure}

Fig.~\ref{fig: experiment4}  presents some examples of comparing populations of planar triangles using our Wasserstein metric. The plot titles on the top state the Wasserstein-type distance calculated with respect to the chosen moving frames, and the point colors (red vs blue) label the Gaussian mixtures. The two right panels display the triangle shapes of these points in these Gaussian mixtures. The top rows present an example where the Gaussian mixtures differ significantly, in terms of means, covariances and weights, while the example in the bottom row has two Gaussian mixtures that are relatively similar: they have the same means and covariances, and differ only in weights. In the first case the Wasserstein-type distances that are relatively large and are relatively stable with respect to choice of moving frame. The distances are naturally smaller in the second example.

\subsection{Comparing Shape Populations of Nanoparticles}\label{sub: nanoparticles}

In this section, we focus on capturing, quantifying, and comparing shapes of silver nanoparticles observed in industrial manufacturing. 
Silver nanoparticles are produced through a solution phase process, leveraging the radiochemistry of electron beam-induced nanoparticle growth (additional information is available in the paper~\cite{woehl-etal:2012}). Over the course of the synthesis, the shapes of these nanoparticles evolve due to chemical reactions such as atomic addition to particles and particle merging. 
This solution phase process is captured using in situ transmission electron microscopy over a span of 62 seconds, with images taken at a rate of one image per second. Each image, on average, displays around 280 silver nanoparticles. The outlines of these nanoparticles are extracted using segmentation methodology presented in~\cite{vo-park:2018}. Each image in the video is pre-processed (including a step which involves removing particles below a certain size threshold) and segmented, returning a set of planar closed curves denoting the outlines of the individual nanoparticles in that frame. The left panel of Fig.~\ref{fig: nanoparticle_clustering_1} shows some examples of extracted contours in imaged frames from the data set.

\begin{figure}
    \begin{tabular}{ccc}
    \centering
    \includegraphics[width=.31\textwidth]{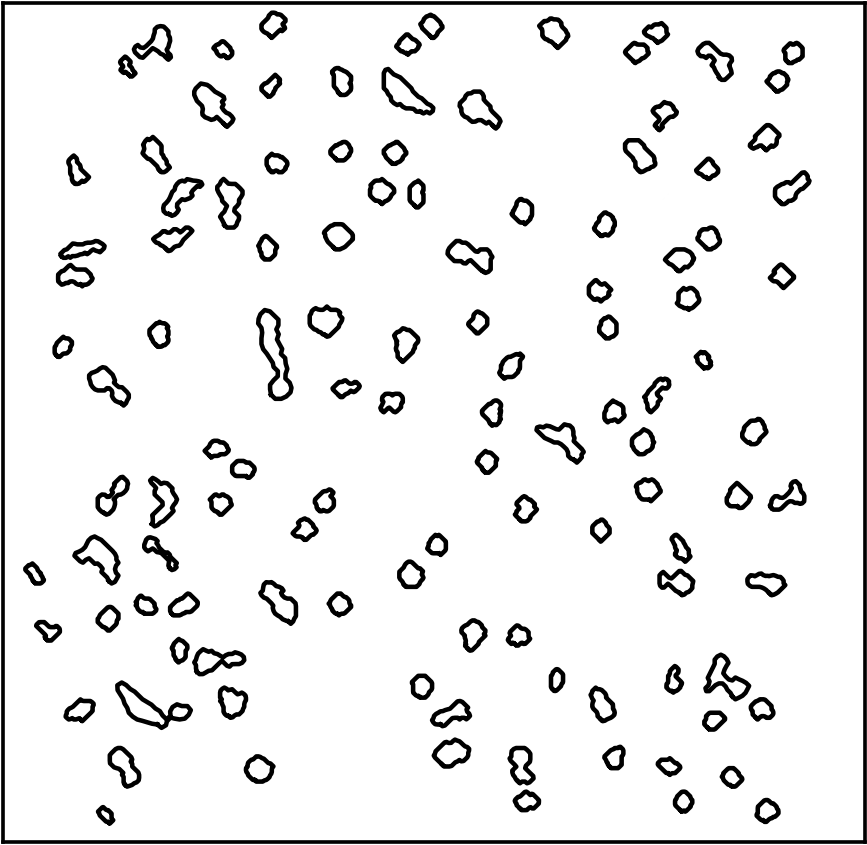} 
    \includegraphics[width=.31\textwidth]{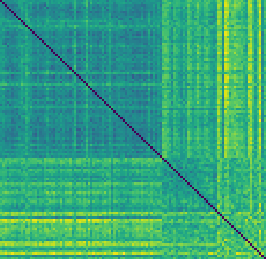}
    \includegraphics[width=.31\textwidth]{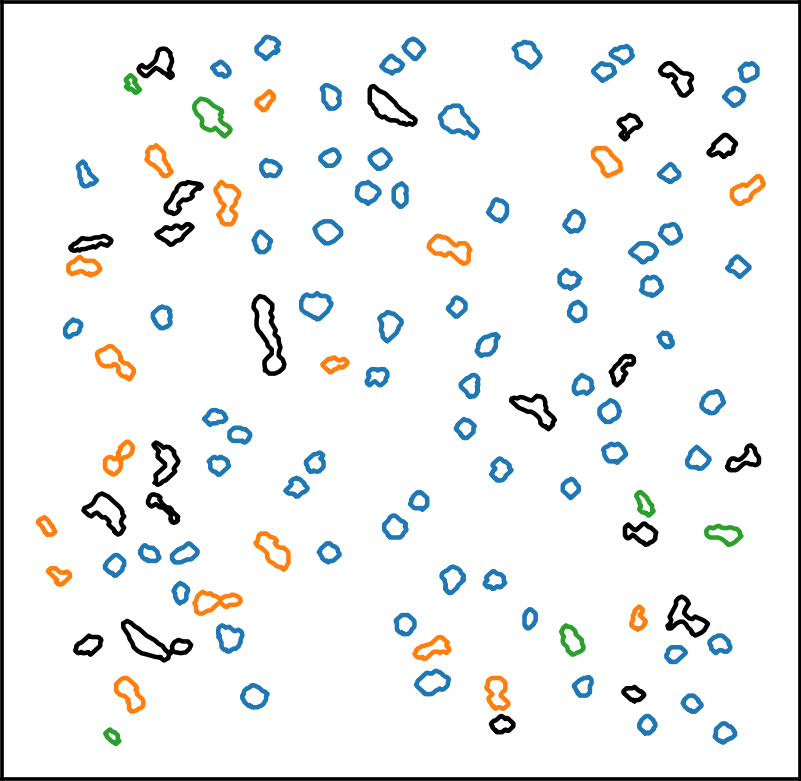}
    \end{tabular}
    \caption{$K$-Modes clustering. Left: Particles in image at $t=38$ before clustering. Middle: pairwise shape distance matrix for all particles that image, sorted by clusters. Right: Particles in image at $t= 38$ colored according to their cluster assignment.}
    \label{fig: nanoparticle_clustering_1}
\end{figure}

In nanomanufacturing, the shapes of nanoparticles are indicators of the material properties. One hypothesis is that constituent nanoparticle shapes can control the resulting material's physical properties. Thus, a vital tool is to model and quantify the particle shape populations associated with individual images and compare them across images. In any image, we treat extracted closed curves as samples from a probability distribution on a shape space, with each curve's location, orientation, and scale treated as nuisance variables. A brief introduction to the shape space ${\cal S}_c = \s^{(2T-1)}/SO(2)$ is presented in Appendix~\ref{appendixA}. Here each contour is represented by an array ${\bf q} \in \real^{2\times T}$ made up of equispaced points on the SRVF curve of the contour. The appendix also defines a shape metric $d_s$, and the computation of  sample statistics (mean and covariance) of a set of shapes under $d_s$. A set of shapes rotationally aligned to their mean can be treated as points on the preshape space, the unit sphere $\s^{(2T-1)}$. On this unit sphere, we take the reference point $p = [1,0,...,0]^T$, and select $F = [[0,1,0,...,0]^T, ..., [0,0,...,0,1]^T]
$ in order to define the moving frame $(p,F)$. We estimate and compare shape distributions with respect to this moving frame.

\paragraph{Estimating Gaussian Mixture Parameters}
Our approach is to model the distribution of shapes in a given video image as a mixture of Gaussians on the preshape space $\s^{(2T-1)}$. There are several steps that make up this approach. 

Given a set of particle contours extracted from a video image, the first step is to cluster them according to their shapes. Fig.~\ref{fig: nanoparticle_clustering_1} shows the process of applying the mode-based clustering process discussed in Section~\ref{sub: general case}. The left panel shows the video image corresponding to time $38$ in the data set, prior to clustering. The middle panel shows the within frame pairwise shape distance matrix, sorted by cluster for these shapes. Green denotes smaller distances, and yellow denotes larger distances. One can see that more than half of the particles fall into the largest cluster. The algorithm automatically selects three clusters and labels the remaining particles as outliers. The right panel shows these particles colored according to their assigned clusters, in blue, orange, and green. The outliers are drawn in black.  

Given a clustering of the shapes in video image $t$, we compute the shape mean $m_k^t$ and tangent space covariance $\Sigma_k^t$ of shapes in cluster $k$, as described in Appendix \ref{appendixA}, and estimate the Gaussian component corresponding to cluster $k$ as $\mu_k^t = N_{\cE}({m}_k^t, {\Sigma}_k^t$). Letting $n_k^t$ be the number of shapes in cluster $k$, and setting $n^t = \sum_k n_k$, we estimate component weights as $w_k^t = \frac{n_k}{n}$. Thus, for each video image, indexed by time $t$, we obtain a Gaussian mixture  $\mu_t = \Sigma_{k=1}^{K_t} w^k N(m^k, \Sigma_k)$.

We then calculate the pairwise Wasserstein-type distances between distributions associated with all images in the video using Eqn.~\ref{eq: mixture distance}. The resulting distance matrix is presented in the leftmost panel of Fig.~\ref{fig:heatmap and mds}.

\begin{figure}
    \includegraphics[width=.99\textwidth]{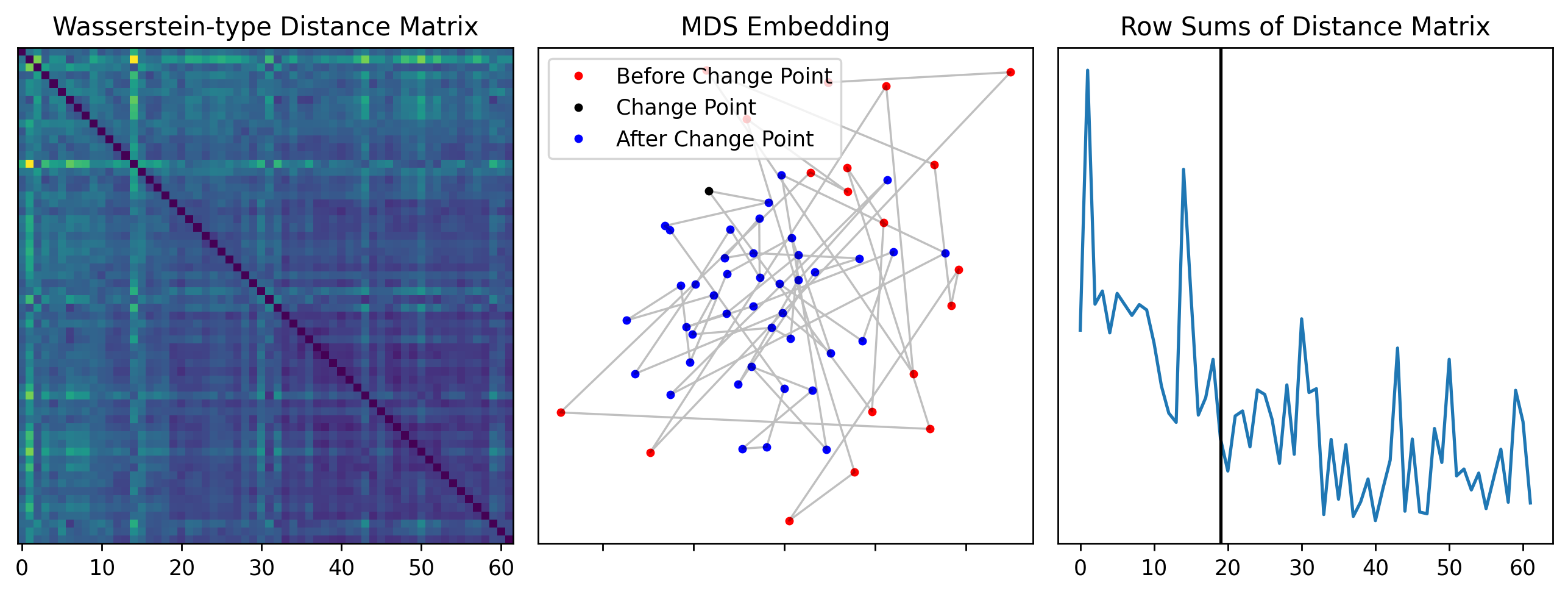}
    \caption{Left: Pairwise Wasserstein-type distance Matrix between time indexed populations, $\mu_t$. Middle: 2D MDS plot based on Wasserstein-type distance matrix, colored by relationship of $t$ to the estimated change point. Right: Sums of rows of the distance Matrix plotted versus $t$, with change point marked by black line.}
    \label{fig:heatmap and mds}
\end{figure}

\paragraph{Change-Point Detection in Time-Series of Shape Populations}
Inspection of the Wasserstein-type distance matrix suggests that the shape distributions in the latter part of the process appear to be closer to each other than to those in the earlier part. In order to test this statistically, we use the {\it E-divisive} procedure for change point detection (\cite{matteson2013nonparametric}). This method is particularly well-suited to our situation, as it only depends on distances between populations, requires minimal assumptions, and provides a straightforward method for testing the hypothesis of no additional change points.

The {\it E-divisive} algorithm is an iterative procedure where candidate change points are selected as the time point which maximizes the two-sample energy statistic (\cite{energy_test}) produced by splitting the data at that time point, and the statistical significance of the candidate change point is inferred on the basis of a permutation test based on the same two-sample energy statistic. The algorithm has several hyperparameters: (1) $R$ the number of permutations, (2) $p_0$ the p-value for each permutation test, (3) $\alpha \in (0,2)$, the power of distance in the test statistics, and (4) $min\_size$, the minimum segment length to be considered for bisection. We applied the {\it E-divisive} algorithm to our Wasserstein-type distance matrix, with parameters $p_0 = 0.0125$, $R = 499$, $min\_size = 12$, and $\alpha = 1$. The algorithm found one statistically significant change point at time point $t=19$, with $p$-value $0.002$. The next candidate change point occurs at time point $t=37$, but is rejected with a p-value of $0.034$, thereby terminating the algorithm.     
These findings lead us to reject the hypothesis of no change points, and provide support for the original observation that the distributions of shapes in the latter part of the manufacturing process differ from those in the earlier part.



\paragraph{ Modeling Shape Dynamics}

\begin{figure}
    \centering
    \includegraphics[width=.525\textwidth]{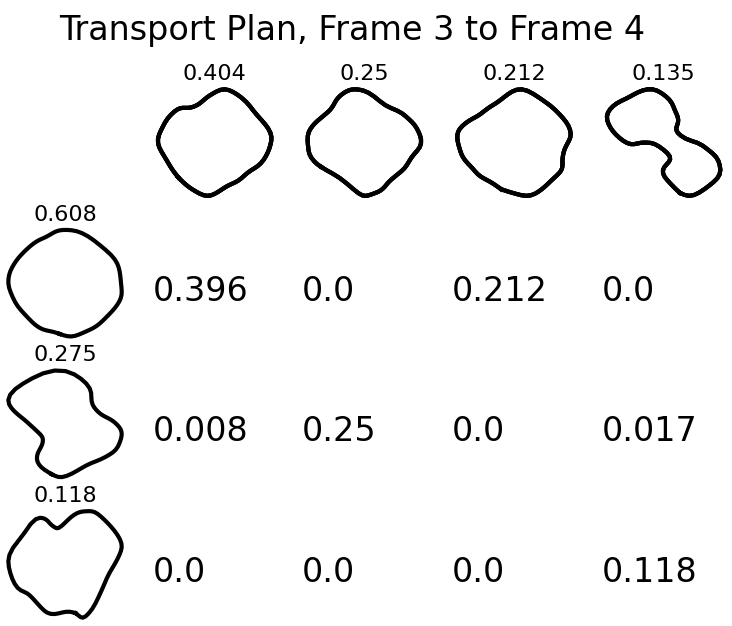}
    \includegraphics[width=.455\textwidth]{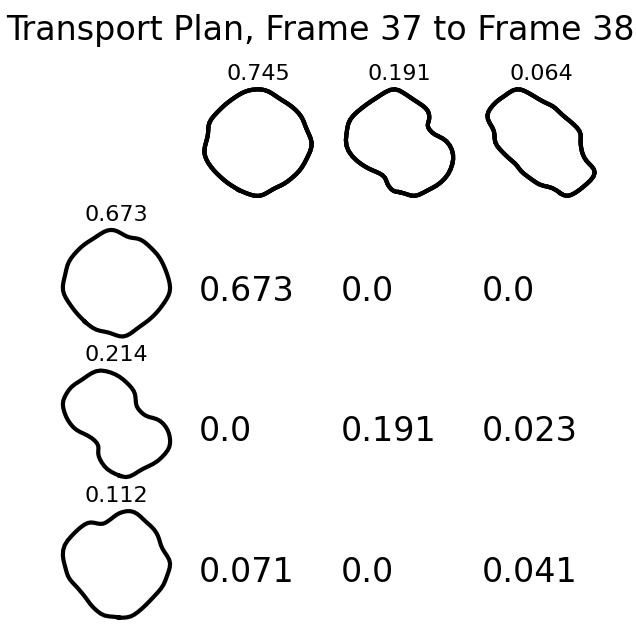}
    \caption{Visual representation of inter-frame population transport plans. The means of Gaussian mixture components are displayed in the margins. Numbers above mean shapes correspond to estimated weight for that component; array of numbers in center of plot are the transport plans, they represent the amount of mass transported between the corresponding marginal components.}
    \label{fig: transport plans}
\end{figure}

In addition to testing \textit{if} a change point exists, it may also be desirable to describe \textit{how} the distributions change over time. The optimal transport plans between shape distributions associated with successive video images provide a natural way to quantify the transitions, and the preshape representation of the means combined with the mixture assumption can make the transport plans simple to interpret. Recall that, for Gaussian mixtures $\mu_{t} = \Sigma_{k=1}^{K_t} w_k N(m_k, \Sigma_k)$ and $\mu_{t+1}= \Sigma_{\ell=1}^{K_t} w_\ell N(m_\ell, \Sigma_\ell)$ associated with images at time $t$ and $t+1$, the optimal transport plan is given by; 

\begin{equation*}
    \pi = \arg \min_{w \in \Pi(w_0,w_1)}  \sum_{k,\ell}^{K_0,K_1} w_{k\ell} W_2^{\cE}(\hat{\mu}_{0}^k,\hat{\mu}_{1}^\ell)^2
\end{equation*}

For example, in the left panel of Figure \ref{fig: transport plans}, we see the optimal transport plan $\pi$ between shape distributions of associated with video images at times $t=3$ and $t=4$. The mean shapes for the source distribution (frame 3) are drawn on the left column, and mean shapes for the target distribution (frame 4) along the top. The weights for the Gaussian mixture components are written above their corresponding mean shape. The optimal transport plan is presented in the rows/columns of the plot. This matrix shows how much mass is transported from each component in the source distribution, to each component in the target distribution. 

These transition matrices can be used to analyze the dynamics of nanoparticle shapes during the manufacturing process. For example, the mass in the cluster with the most circular mean in image $t=3$ (with weight 0.608) ends up being split between two clusters in the transition to image at $t=4$. On the other hand, the mass in the cluster with the most circular mean at $t=37$ (with weight 0.673) is all transported to a single cluster at $t=38$, which also gains most of the mass from another component as well. This dynamic seems to hold for the process in general; transport plans between consecutive frames show a tendency towards distributions with more mass being centered around more circular shapes, especially in the later part of the process. 

\section{Conclusion \& Discussion}\label{sec:conclusion}
This paper develops a framework for representing and comparing populations (probability distributions) on certain nonlinear domains. The domains of interest are trivial vector bundles, with a focus on (finite-dimensional) parallelizable Riemannian manifolds. The populations are represented by mixtures of Gaussians on tangent bundles of these manifolds, and the populations are compared using a convenient expression for a Wasserstein-type distance. This distance is called Wasserstein-type because the search for optimal couplings is restricted to joint mixtures of Gaussians. The paper demonstrates this framework for several examples involving simulated and real data. It uses simulated populations on a unit sphere $\s^2$ to explain how one can compare distributions. The process also involves steps for modeling  populations using mixtures of Gaussians and estimating mixture parameters using clustering methods. 

An important application of this framework is in comparing populations of shapes using image data. This paper uses videos of nanoparticles during a manufacturing process to pursue this application. One associates particles in an imaged frame as samples from a shape population, and compares different image frames using the Wasserstein-type metric between associated shape populations. It further develops a procedure for detecting change point in temporal evolution of shape population during manufacturing. 

In the future, we would like to adapt this framework for solving shape regression problems. In these problems, the shape populations of objects serve as response variables with some Euclidean input variable influencing the outcomes. The goal is to develop statistical models capturing the relationships between input variables and output shape populations. 

    \bibliographystyle{plain}
    \bibliography{refs} 

\appendix

\section{Brief Introduction to the Shape Space of Planar Contours} \label{appendixA}

Here we describe a mathematical representation of shape of closed, planar contours as elements of a finite-dimensional unit sphere. 
Let ${\cal AC}_0$ denote the set of all absolutely-continuous curves 
of the type $\beta: \s^1 \to \real^2$ such that $\beta(t)= 0$ for some $t \in \s^1$. An element $\beta \in {\cal AC}_0$ represents a parameterized planar, closed curve passing through the origin. We are interested in quantifying the shape of $\beta$ in a manner that is invariant to its rotation, translation, scale, and re-parameterization. Taking an elastic approach to shape analysis of curves~\cite{srivastava2016functional,srivastava2011shape}, we represent $\beta$ using its square-root velocity function (SRVF) $q(t) = \frac{\dot{\beta}(t)}{\sqrt{| \dot{\beta}(t)|}}$. The mapping $\beta \mapsto q$ is a bijection from ${\cal AC}_0$ to $\ltwo(\s^1, \real^2)$. Representing a curve $\beta$ by its SRVF $q$ removes the effect of its translations. Further, we rescale $\beta$ to have unit length so that $\|q\|^2 = \mbox{length}(\beta) = 1$. The set of all scaled SRVFs forms a unit Hilbert sphere $\s_{\infty} \subset \ltwo$. To remove reparametrizations, we select a representative shape (this can be same as the distinguished point $p$ on the manifold ${\cal M}$ needed to build a global frame). We reparameterize this representative curve to be arc-length, and then register, through reparameterization, all individual curves (in a given dataset) to this curve. Now we have removed translation, scaling, and reparameterization.  

Next we consider a discretized representation of curves as follows. We sample an SRVF $q$ using $T$ uniformly-spaced sample points $\{t_i \in \s^1 , i= 1, \dots, T\}$ and denote the samples by an array ${\bf q} \in \real^{2 \times T}$ where ${\bf q}_i = q(t_i)$. To ensure unit scale, we rescale the array ${\bf q}$ to have Frobenious norm one and thus we have ${\bf q} \in \s^{2T-1}$. To remove rotation, we use Procrustes alignment in a pairwise fashion as follows. Define the action of $SO(2)$ on $\s^{2T-1}$ as $(O, {\bf q}) = Oq$ and form equivalence classes $[{\bf q}] = \{ O {\bf q}| O \in SO(2)\}$. The shape space of discrete contours, ${\cal S}_c$, is the set of all equivalence classes and denoted by the quotient space $\s^{2T-1}/SO(2)$.
Given any two curves $\beta_1, \beta_2 \in {\cal AC}_0$ and their corresponding discrete SRVFs ${\bf q}_1, {\bf q}_2 \in \s^{2T-1}$, 
$O^* = \argmin_{O \in SO(2)} \| {\bf q}_1 - O {\bf q}_2\|^2$.

The shape metric is then given by: 
$d_s([{\bf q}_1], [{\bf q}_2]) = \cos^{-1}(\inner{{\bf q}_1}{O^*{\bf q}_2})$.
Similarly, given a number of contours $\beta_1, \beta_2, \dots, \beta_n$, one can compute the sample mean of their shapes $[{\bf q}_1], [{\bf q}_2], \dots, [{\bf q}_n]$ according to: 
\[
[\hat{m}] = \argmin_{[{\bf q}] \in {\cal S}} \sum_{i=1}^n d_s([{\bf q}], [{\bf q}_i])^2\ .
\]
An iterative algorithm for finding the minimizer is presented in several places, including~\cite{srivastava2011shape}. In this paper, we use a mode-based procedure~\cite{deng-ICPR:2022,deng-ISBI:2022} to reach estimates of mean shape more efficiently. Once we have computed the sample mean, we can rotationally align individual curves to the mean and express them in a preferred orientation according to: 
\[
{\bf q}_i^* = O_i^* {\bf q}_i, \ \ \mbox{where}\ \ O^* = \argmin_{O \in SO(2)} \| \hat{m} - O {\bf q}_i\|^2\ .
\]
These aligned shapes can be treated as elements of $\s^{2T-1}$ for the purpose of statistical modeling and comparisons. Furthermore, we can compute the shooting vectors ${\bf v}_i^* = \exp_{\hat{m}}^{-1}({\bf q}_i^*) \in T_{\hat{\mu}_n}(\s^{2T-1})$ (on the unit sphere) and define a covariance matrix $\hat{\Sigma} = \frac{1}{n-1} \sum_{i=1}^n {\bf v}_i {\bf v}_i^T $ $\in \real^{(2T-1) \times (2T-1)}$. This gives us a way to represent contour shapes as elements of a finite-dimensional unit sphere, and to define their sample statistics such as means and covariance. One can use these statistics to impose a Gaussian model $\eta = N_{\cal E}(m, \Sigma)$ on $\s^{2T-1}$. 

\section{Mapping Between ${\cal S}_T$ and $\s^2$}\label{hopf equations}

A planar triangle is represented by a matrix $x \in  \real^{3 \times 2}$ or a complex vector $z \in \cC^3$. Let the $i^{th}$ element of $z$ be $z_i = x_{i,1} + j x_{i,2}$. The bijective mappings between the Kendall shape space of triangles ${\cal S}_T$ and $\s^2$ (using Hopf Fibration) are as follows. The forward map from ${\cal S}_T$ to $\s^2$ is given by: 
\begin{eqnarray*}
&& (x_{1,1},\dots, x_{3,2})\ \  \mapsto \left(\theta = \cos^{-1}(\frac{y_3}{r}), \
\varphi =\tan^{-1}(\frac{y_2}{y_1})\right),\ \ \mbox{where}\\ 
&& y_1 = 2(x_{1,2} x_{2,2}+ x_{1,1} x_{2,1}),\ \  y_2 = 2(x_{2,1} x_{2,2} -x_{1,1} x_{1,3}),\ \  y_3 = 1-2(x_{1,2}^2+x_{2,1}^2),\ 
\end{eqnarray*}
and $r = \sqrt{y_1^2 + y_2^2 + y_3^2}$. The backward map from $(\varphi, \theta) \in \s^2$ to ${\cal S}_T$ is given by: 
\begin{eqnarray*}
&&x_{1,1} = \cos(\frac{\psi+\varphi}{2})\sin(\theta/2),\ x_{1,2} = \sin(\frac{\psi+\varphi}{2})\sin(\theta/2),\ \  x_{2,1}=\cos(\frac{\psi-\varphi}{2})\cos(\theta/2),\\
&&x_{2,2} = \sin(\frac{\psi-\varphi}{2})\cos(\theta/2),\ \ x_{3,1} = -(x_{1,1} + x_{2,1}),\ \ x_{3,2} = -(x_{1,2} + x_{2,2})\ ,
\end{eqnarray*}
where $\theta \in [0,\pi]$, $\varphi \in [0,\pi]$ and $\psi \in [0,2\pi]$. 
The angle $\psi$ here is arbitrary and controls the rotation of the resulting triangle. 

\end{document}